\documentclass[conference]{IEEEtran}
\usepackage{amsmath,amssymb}
\usepackage{graphicx}
\usepackage{cite}
\usepackage[hidelinks]{hyperref}

\newtheorem{theorem}{Theorem}
\newtheorem{proposition}{Proposition}
\newtheorem{lemma}{Lemma}
\newtheorem{corollary}{Corollary}
\newtheorem{definition}{Definition}
\newtheorem{assumption}{Assumption}
\newtheorem{remark}{Remark}

\begin{document}

\title{Beyond Provenance: The Economics and Governance of Personalized AI Memory\\}

\author{%
\IEEEauthorblockN{Jianan~Chen}
\IEEEauthorblockA{Purdue University\\
{West Lafayette, IN, USA}\\
chen3873@purdue.edu}
\and
\IEEEauthorblockN{Yufei~Qin}
\IEEEauthorblockA{Princeton University\\
Princeton, NJ, USA\\
{yq9378@princeton.edu}}
\and
\IEEEauthorblockN{Yaosen~Lin}
\IEEEauthorblockA{University of California San Diego\\
La Jolla, CA, USA\\
{yal003@ucsd.edu}}
}

\maketitle

\begin{abstract}
Personalized AI memories, defined as persistent decision-relevant states distilled by platforms from long-run human–AI interaction, are increasingly portable, provenance-certifiable, and tradable. Provenance answers the question of origin, but it does not settle who holds which rights over a memory, nor how wider deployment affects its future supply. We develop a unified model in which memory use is non-rival, while refresh, the maintenance of continued validity, is relationally co-produced by the represented individual and the platform. Even assuming perfect provenance and complete information, the model still yields a set of inherent limitations. First, there is an extraction ceiling: decision value is bounded by lived experience. Second, a rights-separation theorem shows that no budget-balanced ownership assignment can make both co-producers residual claimants, and that control rights and cash-flow rights operate on distinct margins. Third, a depletion result indicates that reach-expanding policies such as portability, interoperability, and competition raise current access but shrink the long-run memory stock. Fourth, we offer a formal comparison between memory ownership as it ought to be (subject-held) and as it is (platform-held); the welfare ranking reduces to two measurable quantities, both dominated by interior renewable-claim bundles. Finally, endogenous-quality and market-formation results show that uncompensated markets pollute their own raw material, that transparency and compensation are complements, and that the de facto platform-ownership regime is self-perpetuating. On the design side, the analysis points to separated provenance, performance, and permission credentials; compensation indexed to quality-deployment; and portability and salience, rather than nominal ownership, as the binding policy margins.
\end{abstract}

\begin{IEEEkeywords}
AI agents; personalized memory; markets for information; data governance;
property rights; consent; certification.
\end{IEEEkeywords}

\section{Introduction}\label{sec:intro}

Personalized AI memory is shifting from a product feature to a persistent asset. Technical work has increasingly converged on agent architectures that accumulate, consolidate, and retrieve long-term memory across sessions \cite{sumers2024coala,hu2025memory}, with recent surveys organizing the literature by memory forms, functional roles, and write/read dynamics \cite{hu2025memory}. One vision paper argues that the resulting artifacts should be managed as human-centric assets rather than as agent-side state \cite{pan2026memoryasset}. Marketplace infrastructure already exists: recent work binds a memory artifact to verifiable computational provenance and adds a marketplace layer for listing, transferring, and governing certified artifacts, where certification attests to authenticity, effort-backed production, and compatibility with the execution context \cite{li2026infrastructure}. Parallel efforts specify portability protocols for moving memory across heterogeneous agents \cite{portablememory2026}. Meanwhile, legal scholarship has begun to ask who owns a simulated identity built from a person's data \cite{jurcys2026twin}.

What makes this class of memory economically distinctive is not that the data pertains to a person, but what a buyer actually obtains. Purchasing a single fact saves one lookup; purchasing a representation of how a person weighs tradeoffs alters multiple future decisions at once. This concentration of decision relevance is both the source of value and the locus of exposure.

The urgency is not speculative. Regulators have already demonstrated the cost of improvised governance over asset classes derived from people. In the 23andMe bankruptcy, deletion rights and oversight for genetic data sourced from millions of individuals were negotiated by a coalition of state attorneys general only after the asset had entered the estate. Personalized memory follows the same trajectory, with two features that complicate governance: it is continuously re-created rather than collected once, and the person it describes continues to produce it.

The puzzle is that verified effort is not economic value. Provenance answers where a memory came from, but it does not answer whether the artifact faithfully represents its subject, whether it suits the buyer's task and host model, whether its extraction and sale were authorized, or whether its accuracy harms the person it describes. The point is sharper than a general call for more metadata: the strongest certification currently proposed attests that an artifact is \emph{effort-backed} \cite{li2026infrastructure}, yet verified production effort is precisely what does not imply decision value. A memory may be authentically and expensively produced from a real interaction history and still be worthless to a buyer; conversely, it may be highly valuable to a buyer while imposing uncompensated exposure on its subject.

Behind this observation lies a public debate that structures our inquiry. One side holds that memory \emph{ought} to belong to the person it represents; the legal case for treating an AI twin as its subject's property rests on the classical criteria of ownership:a defined asset, exclusive control, and demonstrable value \cite{jurcys2026twin}. The other side counters that memory \emph{is} held by the platforms that extract it, because subjects generally lack the technical and behavioral means to retain it. The debate is typically framed as a choice between these two answers. We show that they are, in fact, two poles of a single institutional space, each undermining distinct dimensions of the asset's production, and that the efficient arrangement is neither.

We model a persistent, decision-relevant state whose continued validity is co-produced by a natural person's feedback and a platform's engineering investment, and which can be redeployed at negligible copying cost. To isolate the structural implications, we grant the marketplace every condition its builders aim to secure: provenance is \emph{perfect} and information is \emph{complete}. Consequently, all failures derived below persist even under perfectly functioning infrastructure and cannot be dismissed as artifacts of imperfect verification or user confusion.

The paper rests on a single core mechanism: \textbf{use can be non-rival, while refresh is relationally supplied.} Copying a memory costs nothing and consumes no resources. Maintaining its truth, however, is costly: it requires continuing input from the person it describes, who bears exposure from every deployment and whose incentive to keep supplying that input is precisely what the rights bundle determines. The asset is a non-rival good with a rival maintenance technology, and nearly every result follows from taking that combination seriously.

The paper yields four main contributions.

\begin{itemize}
\item \emph{Conceptual.} We build on existing technical taxonomies of memory rather than proposing new ones, and supplement them with the economic dimensions that any market mechanism must accommodate: use-indexed fidelity, deployment specificity, person linkage, an entanglement index, and rights treated as a designable institution rather than as a definitional premise (Section~\ref{sec:model}).

\item \emph{Theoretical.} Five results, all obtained under perfect provenance and complete information. First, an \emph{extraction ceiling} bounds decision value by the lived interaction history, with equality only when extraction is a sufficient statistic (Proposition~\ref{prop:ceiling}); hence the supply of high-fidelity memory is ultimately constrained by experience that someone must live. Second, a \emph{rights-separation theorem} shows that no budget-balanced assignment of cash-flow rights renders both co-producers residual claimants, so neither individual nor platform ownership implements first-best refresh; moreover, control rights and cash-flow rights operate on distinct margins and cannot substitute for each other (Theorem~\ref{thm:rights}). Third, a \emph{depletion result} establishes that policies widening reach, namely portability, interoperability, and competition, raise current access but shrink next-period new memory and the long-run stock (Theorem~\ref{thm:headline}); nonrivalry of copying therefore does not imply monotone welfare gains from broader deployment. Fourth, an \emph{ought–is resolution} shows that subject ownership and platform ownership are twin corners of the rights space, each extinguishing one supply margin because the non-owning party supplies no refresh; their welfare ranking reduces to two measurable quantities: exposure intensity and the ratio of subject to platform refresh productivity, and both corners are dominated by interior bundles (Proposition~\ref{prop:corners}). The corner observed in practice is moreover self-perpetuating, as the platform's equilibrium play sustains precisely the low portability and low salience that make subject control infeasible. Fifth, a \item \emph{Computational.} We construct a synthetic memory economy in which every estimand has a known ground truth, enabling validation of the paper's estimators at an interior operating point rather than merely proposing them in the abstract (Section VIII). In practice, the computational exercise amended rather than merely illustrated the theory. The Gaussian entanglement frontier turned out to be an information-saturation limit rather than a separate identity; the second-best scope haircut exists only where the subject's refresh margin is live; and the productivity ratio underlying the coverage prescription is systematically attenuated by measurement in high-fidelity systems.

\item \emph{Institutional.} A failure-instrument matrix in which each row carries a theorem pointer, together with a design rule for which the model supplies three independent reasons: provenance, performance, and permission credentials must remain separate rather than fused into a single composite score (Section~\ref{sec:governance}).

\item \emph{Computational.} A synthetic memory economy in which every
estimand has a known ground truth, so that the paper's estimators can be
validated at an interior operating point rather than merely proposed
(Section~\ref{sec:empirics}). It amended the theory rather than illustrating
it. The Gaussian entanglement frontier turned out to be an
information-saturation limit rather than an identity, the second-best scope
haircut exists only where the subject's refresh margin is live, and the
productivity ratio behind the coverage prescription is systematically
attenuated by measurement in high-fidelity systems.
\end{itemize}

The rest of the paper proceeds as follows. Section~\ref{sec:related} surveys the relevant literatures by the module each supplies. Section~\ref{sec:model} presents the model; Section~\ref{sec:results} establishes the main results; Section~\ref{sec:endog} endogenizes fidelity and participation; Section~\ref{sec:optimal} solves for the optimal rights bundle; Section~\ref{sec:governance} translates the results into institutional instruments; Section~\ref{sec:empirics} measures, calibrates, and tests them; Section~\ref{sec:limits} discusses the limits; Section~\ref{sec:conclusion} concludes.

\section{Related Work}\label{sec:related}

We organize related work by the module each literature supplies to our model, and indicate where extrapolation from that work stops. Since our object lies at the intersection of several literatures that rarely interact, the point of this section is less to survey than to identify which piece of machinery comes from where, and which question none has yet posed.

On \emph{agent memory: architectures and taxonomies}, cognitive-architecture work distinguishes working from long-term memory and subdivides the latter into episodic, semantic, and procedural forms \cite{sumers2024coala}; recent surveys organize the field by three joint lenses: \emph{forms} (token-level, parametric, latent), \emph{functions} (factual, experiential, working), and \emph{dynamics} (formation, evolution, and retrieval) \cite{hu2025memory}. These provide our object's coordinates, and we adopt their boundaries: externally persisted, cross-session state, distinct from model weights, from a static retrieval corpus, and from transient context. Within those coordinates, our focal object occupies a single cell: it is cross-session and interaction-inferred in origin, externally persisted in separable form, semantic-preference or procedural-policy in function, and continuously consolidated and retrieved for future action.  ``Preference/policy memory'' is our economic label for that cell, not a proposed cognitive category; the personalized-agent literature supplies the existence proof that memory of this kind steers real actions taken on a user's behalf \cite{personaagent2026}. We take these architectures as given and do not propose new ones. What this body of work does not attempt, and what we do, is the market and governance analysis of the artifacts it produces.

On \emph{memory as an asset and verifiable marketplaces}, two strands supply the phenomenon we study. The first argues for moving from agent-centric to human-centric memory management and treats memory as an asset with a subject \cite{pan2026memoryasset}. It supplies the asset vision and the object, but it does not model price formation, private information, equilibrium, or welfare, and its single-owner framing sits uneasily with co-production; we therefore treat ownership as an outcome to be derived rather than a premise to be assumed. The second supplies working infrastructure: provenance binding, a marketplace layer for listing, transferring, and governing certified artifacts, where certification attests to authenticity, effort-backed production, and compatibility with the execution context \cite{li2026infrastructure}, alongside companion protocols for verified transfer across heterogeneous agents \cite{portablememory2026}. We retain their modules (lineage attestation, compatibility classes, and credential separation) and our identification strategy goes further by treating their engineering as \emph{perfectly successful}. The gap we exploit is that their certification targets production effort; much of what follows argues why verified effort and decision value diverge even when verification works perfectly. Portability work is more than adjacent: the parameter our model uses to represent exit is precisely what such protocols would move, and one of our results explains why no incumbent supplies it voluntarily.

On \emph{ownership of a simulated identity}, the normative pole of our organizing question has been argued directly: individuals should hold moral and legal ownership of the AI twins built from their data, on a human-centric property framework requiring a defined asset, exclusive control, and demonstrable value, with an accompanying policy program of private-by-default governance and portability standards \cite{jurcys2026twin}. That analysis is legal and conceptual, and offers no price formation or welfare analysis; that is where we begin. Our engagement with it is two-sided and, we believe, favorable to that literature on the point that matters most. On one hand, we find that subject ownership, even taken literally, is itself only second best: it internalizes exposure fully but monopolizes deployment and extinguishes the platform's engineering margin, in exactly the way that platform ownership extinguishes the subject's refresh margin. The normative target, accordingly, is a claim structure rather than an owner. On the other hand, that literature's \emph{portability} prescription is vindicated by our model, but for a reason its own framework does not supply: portability and salience, rather than the name on the title, turn out to be the binding margins, and no equilibrium actor provides them. Reform that reassigns ownership while leaving those two conditions in place changes neither the scale of deployment nor the flow of refresh.

In \emph{data economics}, nonrivalry and the welfare stakes of data control are the foundation we build on \cite{jonestonetti2020}, and the externalities that one person's transaction imposes on others provide the machinery behind our market-formation results \cite{acemoglu2022toomuch}. Two departures matter. First, we separate the \emph{subject} from the \emph{seller}; the person described need not be the party selling, and that separation is the wedge our rights results exploit. Second, against the thin-market presumption that the data literature usually carries, we find that identity-entangled memory markets tend to over-thicken at scale, not unconditionally, but because the positive thickness externality saturates while relational leakage does not, and because correlation is strongest precisely among similar individuals. Nonrivalry of copying is where our object starts rather than where it ends, because the maintenance technology is rival even when use is not.

From \emph{property rights and team incentives}, our central rights result builds on two classic frameworks from the property-rights and team-incentives literature: team moral hazard with budget balance \cite{holmstrom1982teams} and the property-rights approach to non-contractible investment \cite{hartmoore1990}. The memory-specific content fills the roles as follows: the two investors are a natural person and a platform; the investments are lived experience and engineering, respectively; and the asset itself is a representation of one of the investors. That last feature has no analogue in the classical setting, and it is precisely what gives exposure its negative sign in the subject's stake. As a consequence, expanding deployment drives the two margins in opposite directions.

From \emph{information, certification, and disclosure}, we draw on three strands: hidden quality \cite{akerlof1970}, the incentive of a monopoly certifier to coarsen its reports \cite{lizzeri1999}, and the disclosure-versus-expropriation trade-off in selling ideas \cite{antonyao1994,antonyao2002}. A lemons application to memory would be true without being memory-specific, and our main results concern distortions that survive symmetric information entirely. The one place this literature becomes load-bearing is certification, where the certifier's proper economic role is not to substitute for the subject's consent but to render fidelity contractible, thereby indexing compensation to measured quality.

Finally, among \emph{adjacent asset classes}, we draw three reference points: software economics for copying, versioning, and licensing; trade secrecy for value-through-exclusivity; and reputation theory supplies a living, transferable asset whose top holders rationally under-maintain \cite{tadelis1999,boardmeyertervehn2013}. The last of these is the closest existing analogue to our refresh margin, and it is also why freshness recertification outperforms a one-time badge. Prompt markets are the nearest neighbour, but serve as a boundary rather than a template: a prompt is a control input whose value need not be a function of history, whereas memory is by construction a statistic of realized interaction. None of these is a synonym for personalized memory, and none transfers wholesale. What is new is endogenous, persistent, person-linked \emph{state} that its subject continues to produce.

\section{Model}\label{sec:model}

\subsection{Problem statement and focal object}\label{sec:problem}

We do not ask what a memory file is worth, nor do we compress the research
question into ``should the individual or the platform own the memory.'' We ask:
for a persistent decision-relevant state whose \emph{continued validity} is
co-produced by a natural person's feedback and a platform's engineering
investment, and which can be redeployed at negligible copying cost: how should
deployment rights, cash-flow rights, update rights, and exit rights be
configured, and how do these configurations change the formation, the
traded scale, and the long-run stock of memory?

The model contains two independent wedges and one aggregation bridge:
(i) a \emph{refresh-incentive wedge}, because future quality is co-produced
but neither party's input is fully contractible; (ii) a \emph{deployment
wedge}, because digital states are reusable, but each deployment creates
buyer value while imposing use-specific exposure on the represented person;
and (iii) a \emph{reach bridge}, because market thickness enters only by
expanding license reach, transmitting the deployment wedge into refresh and
long-run stock. The paper's master mechanism is that \textbf{use can be
non-rival, while refresh is relationally supplied.}

The subject $i$'s latent preference/decision rule at time $t$ is
$\pi_{it}(a\mid x)$ for contexts $x\in\mathcal X$ and actions
$a\in\mathcal A$; interactions generate a history $H_{it}$; a platform's
extraction technology forms the state $s_{it}=\mathcal E_h(H_{it})$, which a
host $h$ loads through a behavioral kernel $K_h(a\mid x,s_{it})$. The focal
class is externally persisted, cross-session, \emph{identity-entangled}
preference/policy memory used to predict or act on behalf of a specific
natural person.

\paragraph{Scope.} Most of what this paper excludes is excluded by
Definition~\ref{def:artifact} rather than by separate stipulation, which is
one reason to state the object formally. \emph{Persistence} excludes
transient working memory and in-context state; \emph{Separability} excludes
parametric memory inseparable from model weights, and is therefore also the
technical precondition for any of this to be tradable at all; \emph{Person
linkage} $L_{i}(m)>0$ excludes pure public facts, which have subjects only in
a trivial sense. Two further exclusions are deliberate rather than implied.
Raw interaction archives that cannot be meaningfully separated into artifacts
are outside the unit of trade of Definition~\ref{def:license}. Memory
describing two or more subjects raises an authorization problem that one
person's consent cannot solve; Theorem~\ref{thm:formation} supplies the
machinery for its externality, but the full rights allocation is companion-paper
scope. None of these is excluded as unimportant.

\paragraph{Reader's map.} Section~\ref{sec:definitions} defines the objects;
Sections~\ref{sec:assumptions}--\ref{sec:environment} state assumptions and
the economic environment; Section~\ref{sec:refresh-subgame} characterizes the
refresh subgame used throughout. Section~\ref{sec:results} develops the main
results; Section~\ref{sec:endog} endogenizes fidelity (pollution) and participation
(formation); Section~\ref{sec:governance} translates the results into instruments; Section~\ref{sec:empirics}
gives the measurement interface.

\subsection{Formal definitions}\label{sec:definitions}

\begin{definition}[State decomposition]\label{def:state}
The decision-relevant state of person $i$ at $t$ is
$\omega_{it}=(\omega^{c}_{t},\,\omega^{g}_{k(i)},\,\omega^{p}_{it})$, where
$\omega^{c}$ is a \emph{common} component, $\omega^{g}$ a \emph{class}
component for $i$'s compatibility class $k(i)$, and $\omega^{p}$ a
\emph{personal} component following
\begin{equation}\label{eq:drift}
\omega^{p}_{it}=\rho\,\omega^{p}_{i,t-1}+\varepsilon_{it},
\qquad \varepsilon_{it}\sim\mathcal N(0,\sigma^{2}_{\varepsilon}),\;
\rho\in(0,1),
\end{equation}
with within-class cross-subject correlation
$\mathrm{Corr}(\omega^{p}_{it},\omega^{p}_{jt})=\varrho\ge 0$.
\end{definition}

\begin{remark}\label{rem:coordinates}
One decomposition yields all four coordinates that would otherwise be
introduced separately: $(\omega^{c},\omega^{g})$ govern transferable value;
$\omega^{p}$ governs personalization value and subject exposure; the weight of
$\omega^{g}$ relative to $\omega^{c}$ is deployment specificity; $\rho$ is
staleness; $\varrho$ is the relational externality.
\end{remark}

\begin{definition}[Memory artifact]\label{def:artifact}
Let $h^{t}_{i}=(x_{is},p_{is},a_{is},y_{is})_{s\le t}$ be the interaction
history generating the $\sigma$-algebra $\mathcal F^{h}_{it}$. An object
$m=(\varphi,\Pi,\ell)$ is a \emph{personalized memory artifact} of subject $i$
iff:
\begin{itemize}
\item (\emph{Derivation}) $\varphi=\mathcal E(h^{t}_{i})$ for an
$\mathcal F^{h}_{it}$-measurable extraction operator $\mathcal E$: memory
is a statistic of realized interaction, which separates it from a prompt, a
control input $p_{t}$ whose value need not be a function of history;
\item (\emph{Persistence}) $\varphi$ remains retrievable after the generating
episode closes (excluding working memory / transient context);
\item (\emph{Separability}) there is a decoder $\mathsf{Ret}$ such that
$\pi_{\theta'}(\cdot\mid x,\mathsf{Ret}(\varphi,x))$ is well defined for hosts
$\theta'\neq\theta$ (excluding weight-inseparable parametric memory; this is
the technical precondition of tradability);
\item (\emph{Efficacy}) $\exists x$ with
$\pi_{\theta}(\cdot\mid x,\mathsf{Ret}(\varphi,x))\neq
\pi_{\theta}(\cdot\mid x,\mathsf{Ret}(\varnothing,x))$;
\item (\emph{Person linkage}) $L_{i}(m):=I(\omega^{p}_{i\cdot};\varphi)>0$,
where $I$ denotes mutual information.
\end{itemize}
$\Pi(m)$ is the lineage DAG recording the generative structure
$(h^{t}_{i},\mathcal E)$; a \emph{provenance credential} $\mathsf{Cert}_{\Pi}$
is a proof system about $\Pi$; \emph{perfect provenance} is the strongest case
in which $\mathsf{Cert}_{\Pi}$ reveals the full truth of
$(h^{t}_{i},\mathcal E)$. $\ell$ is the scoped license defined below.
\end{definition}

\begin{definition}[Use-indexed fidelity]\label{def:fidelity}
For a use $u=(A_{u},\mathcal L_{u},\mu^{u}_{0})$ (action set, loss, prior),
write $V_{u}(\cdot):=\mathbb{E}[\mathcal L_{u}(a^{*}(\cdot),\omega)]$ for the
expected loss of the optimal action under the stated information, with
$V_{u}(\mu^{u}_{0}\mid\varphi)$ its posterior counterpart, and assume the
nondegeneracy condition $V_{u}(\mu^{u}_{0})>V_{u}(\omega)$ (the prior is not
already first-best). Define
\begin{equation}\label{eq:qu}
q_{u}(m):=\frac{V_{u}(\mu^{u}_{0})-V_{u}(\mu^{u}_{0}\mid\varphi)}
{V_{u}(\mu^{u}_{0})-V_{u}(\omega)}\in[0,1],
\end{equation}
the normalized decision value of $m$ in use $u$: $0$ = no better than the
prior, $1$ = as good as knowing $\omega$. Fidelity is thus a \emph{derived,
measurable, use-dependent} quantity, not a free scalar. For a buyer $b$ with
use $u_{b}$ we write $v_{bm}:=q_{u_{b}}(m)$ for $b$'s normalized value of
artifact $m$.
\end{definition}

\begin{definition}[Deployment specificity]\label{def:spec}
For a family of uses $\mathcal U'$ (across buyers, tasks, hosts),
$\sigma(m;\mathcal U'):=1-\min_{u\in\mathcal U'}q_{u}(m)\,/\,
\max_{u\in\mathcal U'}q_{u}(m)\in[0,1]$.
\end{definition}

\begin{lemma}[No unconditional global quality score]\label{lem:no-score}
If there exist buyers $b,b'$ and artifacts $j,k$ with $v_{bj}>v_{bk}$ and
$v_{b'j}<v_{b'k}$, then no buyer-independent scalar score $S(m)$ ranks $j,k$
correctly for both buyers.
\end{lemma}

\begin{IEEEproof}
Immediate: any strict ranking by $S$ contradicts one buyer; equality ranks
neither. With Definition~\ref{def:fidelity}, crossing rankings are generic
whenever $\sigma>0$.
\end{IEEEproof}

\begin{corollary}[Staleness is derived]\label{cor:staleness}
Under Definition~\ref{def:state}, for memory generated at $t$ and used at
$t+k$, the personal component of $q_{u}$ decays by the factor $\rho^{2k}$
(exactly, in the Gaussian case). The depreciation rate is an implication of
preference drift, not an assumption.
\end{corollary}

\begin{remark}[Bridge to the reduced form]\label{rem:delta-bridge}
In the law of motion \eqref{eq:lom} below, the personal-component contribution
to depreciation is $\delta_{p}=1-\rho^{2}$; the reduced-form $\delta$
additionally absorbs class/common drift and technological obsolescence, and is
therefore \emph{class-specific and calibratable}, not free.
\end{remark}

\begin{definition}[Use trichotomy and exposure]\label{def:exposure}
Partition uses by their effect on subject $i$:
$\mathcal U=\mathcal U^{\mathrm{self}}\cup\mathcal U^{\mathrm{other}}
\cup\mathcal U^{\mathrm{adv}}$ (acting for $i$; general lessons for others;
extraction, imitation, or discrimination against $i$). Let
$\mathrm{adv}:\mathcal U\to\mathcal U^{\mathrm{adv}}$ map each use to the
adversarial decision problem sharing its sufficient statistics (the best
inference/imitation attack implementable from what the deployment reveals).
Subject exposure from deploying $m$ in use $u$ is
\begin{equation}\label{eq:harm-micro}
H_{i}(m,u)=\lambda_{i}\cdot q_{\mathrm{adv}(u)}(m),
\end{equation}
i.e., harm is proportional to the decision value the deployment hands to an
adversarial decision problem, with private sensitivity $\lambda_{i}$. That
harm rises with fidelity is therefore not assumed: it holds exactly when the
buyer's and the adversary's decision problems share sufficient statistics.
\end{definition}

\begin{definition}[Entanglement coefficient]\label{def:kappa}
Let $u_{B}$ be the buyer use, $u_{S}$ the adversarial use, and $\mathcal G$
the set of garblings of $\varphi$. Define
\begin{equation}\label{eq:kappa}
\kappa^{\star}(m;u_{B},u_{S})
:=1-\frac{\sup_{z\in\mathcal G_{0}}\,q_{u_{B}}(z(\varphi))}
{q_{u_{B}}(\varphi)}\in[0,1],
\end{equation}
where $\mathcal G_{0}:=\{z\in\mathcal G:\ q_{u_{S}}(z(\varphi))=0\}$ is
the set of zero-exposure garblings, with the convention $\sup\varnothing:=0$ (if no zero-exposure garbling exists,
$\kappa^{\star}=1$). $\kappa^{\star}$ is the share of buyer value that
\emph{cannot} be obtained without subject exposure. $\kappa^{\star}=0$: a
zero-exposure, full-value transformation exists (anonymization, aggregation,
synthetic personas have full purchase). $\kappa^{\star}=1$: any positive buyer
value entails positive exposure. \emph{No privacy technology can improve
the frontier; only institutions remain.} This is the formal content of
``identity-entangled memory,'' a continuous and measurable index, orthogonal
to $\sigma$.
\end{definition}

\begin{remark}[Robust version]\label{rem:kappa-eps}
The exact-zero constraint is knife-edge; the robust object is
$\kappa^{\star}_{\epsilon}$ with constraint
$q_{u_{S}}(z(\varphi))\le\epsilon$, and
$\kappa^{\star}=\lim_{\epsilon\downarrow 0}\kappa^{\star}_{\epsilon}$ whenever
the value of the garbling program is continuous at $\epsilon=0$; empirical
work should report $\kappa^{\star}_{\epsilon}$ at stated $\epsilon$.
\end{remark}

\begin{lemma}[$\kappa^{\star}$ versus identity-replacement estimates]\label{lem:kappa-iota}
Let $\iota(m)=1-V_{b}(m^{-s})/V_{b}(m)$ be the identity-replacement
counterfactual (value retained after naively erasing the subject). Then
$\kappa^{\star}(m)\le\iota(m)$, with equality iff naive erasure is the optimal
zero-exposure garbling.
\end{lemma}

\begin{IEEEproof}
$m^{-s}$ is a feasible element of the constraint set defining
$\kappa^{\star}$; the supremum weakly improves on it.
\end{IEEEproof}

\begin{remark}\label{rem:iota-upper}
Empirical identity-replacement audits therefore \emph{overestimate}
entanglement: smarter garblings may preserve more value. Report $\iota$ as an
upper-bound estimator of the theoretical object $\kappa^{\star}$.
\end{remark}

\begin{definition}[Scoped deployment license]\label{def:license}
The unit of trade is not ``a person's memory'' but
$\ell=(m,b,u,h,[t_{0},t_{1}],k,\text{access mode},\text{audit},
\text{revocation rule})$: a grant to recipient $b$ for specified uses, hosts,
horizon, call budget, and audit/revocation terms. What is non-rival is the
physical occupation of the digital original by one deployment; authorization,
refresh input, compliance, inference cost, and subject exposure are not
thereby zero.
\end{definition}

\begin{definition}[Rights bundle / institution]\label{def:bundle}
An institution is not an owner dummy but
$g=(c,\alpha,\tau,\upsilon,\phi,\rho_{x},a_{d})$: deployment control
$c\in\{P,S,J\}$; subject revenue share $\alpha\in[0,1]$; enforceable
per-quality-deployment compensation $\tau$; update rights $\upsilon$;
functional portability $\phi\in[0,1]$; revocation/recall rules $\rho_{x}$;
audit probability $a_{d}$. ``Individual ownership'' and ``platform
ownership'' are corner combinations. The core theorems endogenize
$(c,\alpha,\tau)$; the remaining components enter feasible sets, outside
options, or extensions.
\end{definition}

\begin{definition}[Prompt--memory boundary]\label{def:boundary}
A current-round prompt $p_{t}$ that only enters the control input is not
memory under Definition~\ref{def:artifact}; if
$s_{t+1}=\mathsf{Upd}(s_{t},p_{t})$ and the state is later retrieved, the
persisted content can become a representation of memory. The same string can
play both roles; the asset object is fixed by causal role and rights
structure, not by file format.
\end{definition}

\begin{definition}[Memory stock and market thickness]\label{def:stock}
With $M_{t}$ the measure of supplying subjects and $\bar\psi_{t}$ their mean
precision contribution, the market-available precision of the common component
is $\Psi_{t}=(\sigma^{2}_{0}+1/(M_{t}\bar\psi_{t}))^{-1}$, increasing,
concave, and bounded by $1/\sigma^{2}_{0}$, where $\sigma^{2}_{0}>0$ is common
noise that aggregation cannot remove. $M_{t}$ is the thickness object entering
the reach bridge $n=N(M)$, $N'\ge0$, used in Sections~\ref{sec:results} and~\ref{sec:endog}.
\end{definition}

\subsection{Assumptions}\label{sec:assumptions}

\begin{assumption}[Perfect provenance benchmark]\label{as:provenance}
All logs, model versions, and generation chains are authentic and certified:
$\mathsf{Cert}_{\Pi}$ is perfect. \emph{(Identification strategy: any failure
below is not about fabrication.)}
\end{assumption}

\begin{assumption}[Complete-information benchmark]\label{as:completeinfo}
Subject and platform know uses and parameters; each buyer knows her own type
and fit. \emph{(Headline results do not rely on consumer error or lemons.)}
\end{assumption}

\begin{assumption}[Joint inputs; effort noncontractibility]\label{as:joint}
Refresh output is co-produced by subject effort $e$ and platform effort $x$;
efforts are noncontractible item-by-item, while realized output is ex post
verifiable \cite{holmstrom1982teams,hartmoore1990}.
\end{assumption}

\begin{assumption}[Functional forms]\label{as:forms}
Quadratic effort costs; installed capabilities $\eta_{S},\eta_{P}>0$;
complementarity $\gamma\ge0$ bounded by the stability condition of
Lemma~\ref{lem:refresh}; buyer types $\theta\sim U[0,1]$ (a normalization;
general regular inverse demand preserves signs).
\end{assumption}

\begin{assumption}[Linear exposure benchmark]\label{as:exposure}
Per unit of quality and per deployment, monetizable exposure is $\lambda\ge0$.
\emph{(Microfoundation: $\lambda$ aggregates
Definition~\ref{def:exposure}'s $\lambda_{i}\,q_{\mathrm{adv}(u)}$ over the
served use distribution of the focal class, and is therefore
use-class-specific. The misrepresentation variant
$\lambda_{P}f+\lambda_{M}(1-f)$ is treated as robustness.)}
\end{assumption}

\begin{assumption}[Scope commitment; quasilinear TU benchmark]\label{as:commitment}
The rights bundle can be committed before refresh efforts; utilities are
quasilinear. \emph{(Monetary IR does not exhaust consent; inalienable-use
constraints enter as hard constraints, Section~\ref{sec:governance}.)}
\end{assumption}

\subsection{Environment}\label{sec:environment}

\paragraph{Players.}
A minimal production unit: subject $S$; platform $P$; a continuum of buyers
with types $\theta\sim U[0,1]$; a social planner as welfare benchmark.

\paragraph{Downstream demand.}
A license is worth $v_{\theta}(Q)=\theta Q$ to type $\theta$, where $Q>0$ is
current effective quality (average focal fit normalized in; in the language of
Section~\ref{sec:definitions}, $Q$ is the class-average use-indexed fidelity
$q_{u}$ delivered to served uses, scaled by the stock). Serving the top
$n\in[0,1]$ share at a uniform price gives $p(Q,n)=Q(1-n)$ and, per unit of
$Q$:
\begin{equation}\label{eq:RCB}
\begin{split}
&R(n)=n(1-n),\qquad C(n)=\tfrac{n^{2}}{2},\\
&B(n)=R(n)+C(n)=n-\tfrac{n^{2}}{2},
\end{split}
\end{equation}
extractable license revenue, buyer surplus, and gross use value. Copying does
not consume $Q$, but the subject bears $H(Q,n)=\lambda Q n$.

\paragraph{Refresh technology.}
The effective stock evolves as
\begin{equation}\label{eq:lom}
Q_{t+1}=(1-\delta)Q_{t}+Y(e_{t},x_{t}),\qquad \delta\in(0,1],
\end{equation}
\begin{equation}\label{eq:Y}
\begin{split}
&Y(e,x)=\eta_{S}e+\eta_{P}x+\gamma ex,\\
&K_{S}(e)=\tfrac{c_{S}e^{2}}{2},\qquad K_{P}(x)=\tfrac{c_{P}x^{2}}{2}.
\end{split}
\end{equation}
Installed capabilities $\eta_{S},\eta_{P}>0$ rule out the zero-effort
coordination equilibrium that an essential-inputs Cobb--Douglas would admit
unconditionally; $\gamma=0$ is the conservative additive benchmark in which
all headline results already hold.

\paragraph{Stakes.}
With private per-quality flows $u_{S},u_{P}\ge0$ and rights $(\alpha,\tau)$,
the capitalized marginal stakes in a unit of refresh are
\begin{align}
w_{S}(n)&=u_{S}+\alpha R(n)+(\tau-\lambda)n, \label{eq:wS}\\
w_{P}(n)&=u_{P}+(1-\alpha)R(n)-\tau n, \label{eq:wP}\\
T(n)&=w_{S}+w_{P}=u_{S}+u_{P}+R(n)-\lambda n, \notag\\
\Omega(n)&=u_{S}+u_{P}+B(n)-\lambda n=T(n)+C(n). \label{eq:Omega}
\end{align}
Prices, revenue shares, and $\tau$ are internal transfers and do not enter
$\Omega$ twice; buyer surplus $C(n)$ is the part no private rights bundle
automatically captures. (Fines paid to third parties reduce $w_{P}$ without
raising $w_{S}$: a separate liability wedge, not conflated with $\tau$.) A
common capitalization factor $\Gamma(\beta,\delta)=\beta/(1-\beta(1-\delta))$
is normalized to one under stationarity; restoring it rescales all flow terms
without changing the comparisons below. Effort-stage payoffs and incremental
welfare are
\begin{equation}\label{eq:payoffs}
\begin{split}
&\pi_{S}=w_{S}Y-\tfrac{c_{S}e^{2}}{2}+F,\qquad
\pi_{P}=w_{P}Y-\tfrac{c_{P}x^{2}}{2}-F,\\
&W=\Omega\,Y-\tfrac{c_{S}e^{2}}{2}-\tfrac{c_{P}x^{2}}{2},
\end{split}
\end{equation}
with fixed transfer $F$ allocating rents but not marginal incentives.

\paragraph{Timing.}
(1) institution $g$ and $F$; (2) scope commitment $n$ by the controller named
in $c$; (3) simultaneous noncontractible $(e,x)$; (4) realization and
verification of $Y$ (lineage perfect; fidelity auditable on pre-registered
tasks); (5) licensing at $p(Q,n)$; (6) deployment, payments, exposure;
(7) transition by \eqref{eq:lom}. Equilibrium: subgame perfection under
rational expectations; bounded attention only in the robustness section.

\paragraph{Participation.}
Buyer IR holds ex post iff $\theta\ge1-n$. Subject and platform IR,
$\Pi_{S}\ge\bar U_{S}$, $\Pi_{P}\ge\bar U_{P}$, can be met by $F$ whenever
joint surplus suffices, but $F$ cannot repair the slopes of
\eqref{eq:wS}--\eqref{eq:wP}. Inalienable uses impose the stronger constraint
$\ell\in\mathcal R_{i}(m)$, which monetary IR does not replace.

\subsection{The refresh subgame}\label{sec:refresh-subgame}

Let $A_{S}=w_{S}/c_{S}$, $A_{P}=w_{P}/c_{P}$,
$\Delta_{\gamma}=1-\gamma^{2}A_{S}A_{P}$.

\begin{lemma}[Unique active refresh equilibrium]\label{lem:refresh}
If $w_{S}>0$, $w_{P}>0$, $\Delta_{\gamma}>0$, the effort subgame has a unique
strictly positive Nash equilibrium
\begin{equation*}
e^{E}=\frac{A_{S}(\eta_{S}+\gamma A_{P}\eta_{P})}{\Delta_{\gamma}},\qquad
x^{E}=\frac{A_{P}(\eta_{P}+\gamma A_{S}\eta_{S})}{\Delta_{\gamma}};
\end{equation*}
in the additive benchmark $\gamma=0$,
\begin{equation}\label{eq:YE}
\begin{split}
&e^{E}=\frac{\eta_{S}w_{S}}{c_{S}},\qquad
x^{E}=\frac{\eta_{P}w_{P}}{c_{P}},\\
&Y^{E}=k_{S}w_{S}+k_{P}w_{P},\qquad
k_{S}:=\frac{\eta_{S}^{2}}{c_{S}},\;\;
k_{P}:=\frac{\eta_{P}^{2}}{c_{P}}.
\end{split}
\end{equation}
If $w_{j}\le0$, participant $j$ corners at zero; interior derivatives are used
only on the active region. The planner's problem replaces both stakes by
$\Omega$ and requires $1-\gamma^{2}\Omega^{2}/(c_{S}c_{P})>0$ for a unique
interior first best.
\end{lemma}

$k_{S}$ and $k_{P}$, the \emph{installed refresh productivities} of subject
feedback and platform engineering, are the two empirical primitives on
which every institutional comparison below turns.

\subsection{Where each object binds}\label{sec:object-map}

No object above is defined for its own sake; each is consumed at a named
point downstream. Definition~\ref{def:state}'s $\varrho$ is the relational
externality driving over-participation in Section~\ref{sec:endog} (formation); $\rho$
delivers Corollary~\ref{cor:staleness} and calibrates $\delta$
(Remark~\ref{rem:delta-bridge}). Definition~\ref{def:artifact}'s
\emph{Separability} and Definition~\ref{def:bundle}'s portability $\phi$ are
the \emph{objective feasibility conditions} for subject control compared in
Section~\ref{sec:corners}: the formal content of ``whether users can in
fact own their memory.'' Definition~\ref{def:fidelity} and
Lemma~\ref{lem:no-score} rule out scalar quality scores and power the
evaluation design of Sections~\ref{sec:extensions} and~\ref{sec:empirics}.
Definitions~\ref{def:exposure}--\ref{def:kappa} microfound the reduced-form
$\lambda$ (Assumption~\ref{as:exposure}) and bound what privacy technology can
substitute for institutions ($\kappa^{\star}$, tested in Section~\ref{sec:empirics}).
Definition~\ref{def:license} fixes the priced unit in
Section~\ref{sec:results}'s licensing stage. Definition~\ref{def:bundle} is
the space in which Sections~\ref{sec:static-wedge} and~\ref{sec:corners}
locate the \emph{de facto} and \emph{normative} corner institutions and
Corollary~\ref{cor:stake} locates the interior optimum.
Definition~\ref{def:boundary} guards the paper's boundary against prompt markets.
Definition~\ref{def:stock} supplies the thickness object $M$ whose reach
bridge $n=N(M)$ closes the loop from Section~\ref{sec:endog} back into
Theorem~\ref{thm:headline}.

\section{Main Results: Non-Rival Use, Relationally Supplied Refresh}\label{sec:results}

This section develops the model's results in three layers, organized by a
question the paper keeps distinct from the model itself: memory ownership
\emph{as it ought to be}, held by the person the memory represents,
versus \emph{as it is}, held by platforms under current technical and
behavioral conditions. Definition~\ref{def:bundle} makes both expressible
as corner institutions in one bundle space, so their comparison is a theorem,
not a stance. Section~\ref{sec:ceiling} bounds what any extraction technology
can deliver. Section~\ref{sec:static-wedge} freezes the stock and isolates
the \emph{deployment wedge}: who controls scope matters, and the direction of
the platform's distortion reverses in the exposure intensity, which recovers,
as a special case, the distortion-reversal regimes of the static monopoly
model. Sections~\ref{sec:rights}--\ref{sec:headline} open the
\emph{refresh-incentive wedge}: no rights bundle makes both co-producers
residual claimants, and because the represented person remains an endogenous
supplier of refresh, broader deployment can raise current access while
shrinking the long-run stock: the paper's headline.
Section~\ref{sec:corners} then confronts the two ownership corners directly:
each destroys the other party's supply margin, their welfare ranking reduces
to two measurable quantities, and \emph{neither} is the normative target.

\subsection{The extraction ceiling}\label{sec:ceiling}

\begin{proposition}[Extraction ceiling]\label{prop:ceiling}
Under Definition~\ref{def:artifact} (Derivation), $\varphi=\mathcal E(h^{t})$
with $\mathcal E$ measurable with respect to $\mathcal F^{h}_{t}$, so
$\omega\to h^{t}\to\varphi$ is a Markov chain. Hence for every use $u$,
$q_{u}(\varphi)\le q_{u}(h^{t})$, with equality iff $\varphi$ is a sufficient
statistic of $h^{t}$ for $u$. No extraction technology, however
sophisticated, can deliver more decision value than the lived interaction
history contains.
\end{proposition}

\begin{IEEEproof}
Data-processing: any decision rule based on $\varphi$ is a garbling of one
based on $h^{t}$; Blackwell's theorem orders the normalized values in
Definition~\ref{def:fidelity}.
\end{IEEEproof}

\begin{remark}\label{rem:ceiling}
Two consequences follow. (i) Supply of high-fidelity memory is ultimately
bounded by \emph{experience}, which must be lived by someone: this is the
microfoundation for the depletion logic of Theorem~\ref{thm:headline}.
(ii) Claims that
extraction quality can substitute indefinitely for interaction quality are
falsified within the model. The result is elementary as mathematics; its role
is to fix the direction of scarcity before any institutional question is
asked.
\end{remark}

\subsection{The static deployment wedge: scope control and distortion
reversal}\label{sec:static-wedge}

Freeze the stock at $Q>0$ (no refresh margin) and let the controller named in
$c$ choose reach $n$ after $(\alpha,\tau)$ are set. This diagnostic benchmark
isolates what deployment control \emph{alone} does; every force in it
survives into the dynamic model, where it is joined by the refresh margin.

\begin{proposition}[Scope divergence]\label{prop:scope}
In the fixed-stock benchmark, the platform-controlled, subject-controlled,
and welfare-maximizing scopes are
\begin{equation}\label{eq:scopes}
\begin{split}
&n_{P}=\Bigl[\tfrac{1-\tau/(1-\alpha)}{2}\Bigr]_{[0,1]},\qquad
n_{S}=\Bigl[\tfrac{1-(\lambda-\tau)/\alpha}{2}\Bigr]_{[0,1]},\\
&n_{W}=[1-\lambda]_{[0,1]},
\end{split}
\end{equation}
maximizing $(1-\alpha)R(n)-\tau n$, $\alpha R(n)+(\tau-\lambda)n$, and
$B(n)-\lambda n$ respectively. In the $\tau=0$, $0<\alpha<1$ special case,
$n_{P}=\tfrac12$ independently of both $\alpha$ and $\lambda$: revenue shares
do not move the platform's scope, and the platform's scope does not respond
to exposure. (For $\tau\neq0$, $n_{P}$ does vary with $\alpha$.)
\end{proposition}

\begin{corollary}[Distortion reversal, unified-model version]\label{cor:reversal}
Set $\alpha=\tau=0$, $u_{S}=u_{P}=0$ (pure platform control, the
\emph{de facto} corner $g_{P}$ of Section~\ref{sec:corners}). Then
$w_{P}(n)-\Omega(n)=\tfrac{n(2\lambda-n)}{2}$, so at any scope the platform's
marginal refresh stake exceeds the social stake iff $\lambda>n/2$; and
$n_{P}=\tfrac12$ versus $n_{W}=1-\lambda$. Hence three regimes:
\begin{itemize}
\item \textbf{Low exposure} ($\lambda<\tfrac14$): the platform under-deploys
and its investment stake falls short of the social stake, the classic
monopoly shortfall on both margins;
\item \textbf{Intermediate exposure} ($\tfrac14<\lambda<\tfrac12$): the
platform's engineering stake already exceeds the social stake while
deployment is still too narrow (over-extraction with under-deployment);
\item \textbf{High exposure} ($\lambda>\tfrac12$): the platform over-invests
on its margin and over-deploys.
\end{itemize}
In all three regimes the subject margin is degenerate: $w_{S}=-\lambda n\le0$,
so $e^{E}=0$. Under pure platform control the subject supplies no refresh
at any exposure level, and total quality can remain underprovided even where
the platform over-invests (whenever $k_{S}$ is large).
\end{corollary}

\begin{IEEEproof}
First-order conditions of the three programs in
Proposition~\ref{prop:scope} (interior cases; corners by the projection). The
stake comparison is direct from \eqref{eq:wS}--\eqref{eq:Omega}.
\end{IEEEproof}

\begin{remark}[Relation to the static monopoly model]\label{rem:reversal}
The earlier draft's Theorem~1 obtains the same three-regime structure in a
model where the platform buys fidelity $q$ at convex cost and person linkage
prices exposure $\lambda q$ per deployment: under-both for
$\lambda/A<1-1/\sqrt2$, over-fidelity/under-deployment in between, over-both
for $\lambda/A>1/2$. Corollary~\ref{cor:reversal} shows those regimes are not
an artifact of that parametrization: they re-emerge inside the unified model
as the fixed-stock limit, with thresholds ($\tfrac14$, $\tfrac12$ under
$U[0,1]$ demand) determined by demand curvature. What the static model could
not see, because it had no subject input, is the third margin: the
represented person's refresh supply is already destroyed in \emph{every}
regime under pure platform control. The regimes discipline rhetoric on both
sides (``markets under-provide quality'' holds only at low exposure;
``platforms strip-mine identities'' only past the thresholds), and the middle
regime implies that neither a pure quality subsidy nor a pure output subsidy
is the right instrument: each moves the other margin adversely.
\end{remark}

\subsection{Rights separation}\label{sec:rights}

We now restore the refresh margin, and the next result says that the
deployment wedge cannot be repaired by reassigning ownership: no
budget-balanced bundle makes \emph{both} co-producers residual claimants,
and control rights and cash-flow rights operate on different margins.

\begin{theorem}[Rights separation]\label{thm:rights}
Assume Lemma~\ref{lem:refresh}'s conditions, a strictly positive first best,
noncontractible efforts, incentives only via marginal claims on realized
refresh, ex post budget balance, and no external subsidy.
\begin{enumerate}
\item[(i)] \textbf{(Team-incentive impossibility.)} First-best efforts
require $w_{S}=\Omega$ \emph{and} $w_{P}=\Omega$; but any budget-balanced
bundle allocating private cash flows and monetizable exposure satisfies
$w_{S}+w_{P}=T=\Omega-C<2\Omega$. Hence no scalar owner and no single
revenue-share rule implements first-best refresh on both margins, even
under perfect price discrimination ($C=0$), since then
$w_{S}+w_{P}=\Omega<2\Omega$.
\item[(ii)] \textbf{(Control--cash-flow non-equivalence.)} With
$\widehat Q(n)=(1-\delta)Q_{t}+Y^{E}(n)$, an interior ex-ante
platform-controlled scope satisfies
$w_{P}'\widehat Q+k_{S}w_{P}w_{S}'=0$, while an interior subject-controlled
scope satisfies $w_{S}'\widehat Q+k_{P}w_{S}w_{P}'=0$: generically different
conditions at the same $(\alpha,\tau)$. By Proposition~\ref{prop:scope}, in
the $\tau=0$ special case varying $\alpha$ leaves $n_{P}$ unchanged while
moving $n_{S}$; varying the controller moves $n$ at fixed $\alpha$. Cash-flow
rights cannot substitute for scope control, nor conversely.
\end{enumerate}
\end{theorem}

\begin{IEEEproof}[Proof sketch]
(i) Decentralized FOCs $c_{S}e=w_{S}Y_{e}$, $c_{P}x=w_{P}Y_{x}$ against
first-best FOCs with $\Omega$; positivity of $Y_{e},Y_{x}$ forces both stakes
to equal $\Omega$, contradicting budget balance; this is Holmstr\"om-type
team moral hazard, with buyer surplus widening the wedge. (ii) Substitute
Lemma~\ref{lem:refresh} into \eqref{eq:payoffs}, differentiate in $n$ using
$\widehat Q'=k_{S}w_{S}'+k_{P}w_{P}'$; corners at $\alpha\in\{0,1\}$ treated
separately. The theorem excludes only budget-balanced marginal-revenue
implementations; verifiable effort, external subsidies, or performance
contracts change the domain.
\end{IEEEproof}

\begin{corollary}[Optimal marginal-stake allocation]\label{cor:stake}
In the additive benchmark, fixing $n$ and total private stake $T>0$ with
$w_{P}=T-w_{S}$, decentralized welfare
$W(w_{S})=k_{S}(\Omega w_{S}-\tfrac{w_{S}^{2}}{2})
+k_{P}(\Omega w_{P}-\tfrac{w_{P}^{2}}{2})$
is uniquely maximized on the budget-balanced set at
\begin{equation}\label{eq:wstar}
w_{S}^{\star}=\operatorname{Proj}_{[0,T]}
\Bigl\{\frac{(k_{S}-k_{P})\Omega+k_{P}T}{k_{S}+k_{P}}\Bigr\},
\qquad w_{P}^{\star}=T-w_{S}^{\star},
\end{equation}
with $\partial w_{S}^{\star}/\partial k_{S}>0$ and
$\partial w_{S}^{\star}/\partial k_{P}<0$ in the interior: measurably higher
subject refresh productivity tilts optimal marginal claims toward the
subject. The unconstrained revenue share implementing $w_{S}^{\star}$ is
$\widetilde\alpha=(w_{S}^{\star}-u_{S}-(\tau-\lambda)n)/R(n)$; it is
implementable iff $\widetilde\alpha\in[0,1]$ with $w_{P}>0$ and IR-feasible
$F$; otherwise the projected $\alpha$ is only the nearest feasible
boundary and does \emph{not} attain $w_{S}^{\star}$.
\end{corollary}

\begin{remark}\label{rem:stake-econometrics}
Corollary~\ref{cor:stake} is the constructive counterpart of
Theorem~\ref{thm:rights}: given that first best is out of reach, the
second-best allocation of marginal claims is an estimable formula in
$(k_{S},k_{P})$, the two primitives the measurement interface of Section~\ref{sec:empirics}
targets. ``Who should hold renewable claims'' is thereby converted from an
ideological question into an econometric one.
\end{remark}

\subsection{Non-rival use, rival refresh}\label{sec:headline}

The headline result combines both wedges; to state it, define the equilibrium
stake elasticities $\varepsilon_{j}:=\partial\ln Y^{E}/\partial\ln w_{j}>0$
for $j\in\{S,P\}$, evaluated on the active region (in the additive benchmark,
$\varepsilon_{S}=k_{S}w_{S}/Y^{E}$ and $\varepsilon_{P}=k_{P}w_{P}/Y^{E}$,
the effort-supply shares).

\begin{theorem}[Non-rival use, rival refresh]\label{thm:headline}
Fix $(\alpha,\tau)$ and work in the active region. Since $Y^{E}$ depends on
$n$ only through the stakes,
\begin{equation*}
\frac{d\ln Y^{E}}{dn}
=\varepsilon_{S}\frac{\alpha(1-2n)+\tau-\lambda}{w_{S}}
+\varepsilon_{P}\frac{(1-\alpha)(1-2n)-\tau}{w_{P}},
\end{equation*}
so broader licensing reduces new refresh iff
\begin{equation}\label{eq:depletion}
\varepsilon_{S}\frac{\lambda-\tau-\alpha(1-2n)}{w_{S}}
>\varepsilon_{P}\frac{(1-\alpha)(1-2n)-\tau}{w_{P}}.
\end{equation}
(The one-directional narrative, in which expansion erodes the subject stake
while raising the platform stake, applies when $w_{S}'<0<w_{P}'$; the
algebra of \eqref{eq:depletion} is valid regardless.) In the additive
benchmark,
\begin{equation*}
\frac{dY^{E}}{dn}
=k_{S}[\alpha(1-2n)+\tau-\lambda]+k_{P}[(1-\alpha)(1-2n)-\tau],
\end{equation*}
and with $\alpha=\tau=0$: $dY^{E}/dn<0\iff k_{S}\lambda>k_{P}(1-2n)$; for
$n\ge\tfrac12$ and $\lambda>0$, $Y^{E\prime}(n)<0$ unambiguously. In steady
state $\bar Q(n)=Y^{E}(n)/\delta$, so
$\operatorname{sign}\bar Q'(n)=\operatorname{sign}Y^{E\prime}(n)$:
\emph{policies that raise reach (portability, interoperability,
competition) can increase current access while reducing the next period's
new memory and the long-run stock.}
\end{theorem}

\begin{IEEEproof}
Chain rule on \eqref{eq:YE} (additive case) and on
Lemma~\ref{lem:refresh}'s closed forms (general case), with
$w_{S}'(n)=\alpha(1-2n)+\tau-\lambda$ and
$w_{P}'(n)=(1-\alpha)(1-2n)-\tau$ from \eqref{eq:wS}--\eqref{eq:wP}; the
steady state solves $\delta\bar Q=Y^{E}$.
\end{IEEEproof}

\begin{corollary}[Access gains do not imply continuation-welfare gains]\label{cor:access}
With $W^{E}(n)=\Omega Y^{E}-\tfrac{k_{S}w_{S}^{2}}{2}
-\tfrac{k_{P}w_{P}^{2}}{2}$ and $\Omega'(n)=1-n-\lambda$: even when
$\Omega'(n)>0$ (a marginal deployment has positive static social value),
$W^{E\prime}(n)<0$ on an open parameter region compatible with
$0<\lambda<\tfrac12$. Nonrivalry of copying is not sufficient for monotone
welfare gains from broader deployment: as long as the represented person
remains an endogenous supplier of refresh, deployment scope enters the future
supply function.
\end{corollary}

\begin{proposition}[Fixed buyouts create no refresh margin; renewable claims
do, with financing caveats]\label{prop:renewable}
On the same active branch (participation, rights, $n$ unchanged), a fixed
buyout price $F$ satisfies
$\partial e^{E}/\partial F=\partial x^{E}/\partial F=0$; an alienating sale
leaving $w_{S}^{\mathrm{sale}}\le0$ induces $e^{\mathrm{sale}}=0$. A
renewable claim $\Delta w_{S}>0$ raises subject effort by
$\Delta e=(\eta_{S}/c_{S})\Delta w_{S}$ always, but its effect on total
refresh depends on financing: if funded by new extractable surplus or outside
funds (locally $w_{P}$ constant), $\Delta Y=k_{S}\Delta w_{S}$ and
$\Delta\bar Q=(k_{S}/\delta)\Delta w_{S}$; if a budget-balanced reallocation
($\Delta w_{P}=-\Delta w_{S}$), then $\Delta Y=(k_{S}-k_{P})\Delta w_{S}$,
which is positive iff the subject is the more productive refresher, and
welfare-improving only when it moves stakes toward
Corollary~\ref{cor:stake}'s $w_{S}^{\star}$. When continuation refresh has
substantial value, the essence of ``memory as an asset'' is not minting a
transferable token but designing a relational contract that keeps paying for
refresh.
\end{proposition}

\subsection{Ownership as it ought to be and as it is}\label{sec:corners}

The public debate the paper addresses has a normative pole (memory
\emph{ought} to belong to the person it represents) and a positive pole
(memory \emph{is}, under current conditions, held by the platforms that
extract it). In the space of Definition~\ref{def:bundle} these are two corner
institutions:
\begin{align*}
g_{S}&=(c{=}S,\;\alpha{=}1,\;\tau{=}0)\ \ \text{(the \emph{ought} corner)},\\
g_{P}&=(c{=}P,\;\alpha{=}0,\;\tau{=}0)\ \ \text{(the \emph{is} corner)}.
\end{align*}
The comparison below is stated in the additive benchmark with
$u_{S}=u_{P}=0$; all closed forms are exact there.

\begin{proposition}[The ought and the is]\label{prop:corners}
\begin{enumerate}
\item[(i)] \textbf{(Symmetric margin destruction.)} Under $g_{P}$,
$w_{S}=-\lambda n\le0$ and $e^{E}=0$; under $g_{S}$, $w_{P}=0$ and
$x^{E}=0$. Each corner converts a co-production technology into single-input
production: whoever does not own supplies no refresh.
\item[(ii)] \textbf{(Scope.)} Because the non-owner's margin is dead, each
owner's ex-ante payoff is
$\pi_{j}=w_{j}(n)(1-\delta)Q_{t}+\tfrac{k_{j}w_{j}(n)^{2}}{2}$, monotone in
the own stake, so the ex-ante scope choice coincides with the fixed-stock
program of Proposition~\ref{prop:scope}. The ought corner internalizes
exposure fully but monopolizes: $n_{S}=\tfrac{1-\lambda}{2}=\tfrac{n_{W}}{2}$,
so deployment is always exactly half the first best, never excessive. The
is corner is exposure-blind: $n_{P}=\tfrac12$, with the reversal regimes of
Corollary~\ref{cor:reversal}.
\item[(iii)] \textbf{(Welfare ranking is an empirical question.)} Corner
welfares are
\begin{equation}\label{eq:corner-welfare}
W_{S}=\frac{k_{S}(1-\lambda)^{4}}{16},\qquad
W_{P}=\frac{k_{P}(1-2\lambda)}{16},
\end{equation}
so the ought corner dominates the is corner iff
$k_{S}(1-\lambda)^{4}>k_{P}(1-2\lambda)$: always when $\lambda\ge\tfrac12$,
and at low exposure iff the subject is (roughly) the more productive
refresher, $k_{S}/k_{P}>(1-2\lambda)/(1-\lambda)^{4}$. The ownership debate
thereby reduces to two measurable quantities: exposure intensity $\lambda$
and the refresh-productivity ratio $k_{S}/k_{P}$.
\item[(iv)] \textbf{(Neither corner is the normative target.)} Fix either
corner's scope $n$ and the implied total stake $T(n)>0$. Whenever
Corollary~\ref{cor:stake}'s optimum is interior
($0<(k_{S}-k_{P})\Omega+k_{P}T<(k_{S}+k_{P})T$, which holds for all
$k_{S}/k_{P}$ in a neighborhood of $1$), reallocating marginal claims from
the corner to $w_{S}^{\star}\in(0,T)$ strictly raises welfare, and
re-optimizing $n$ raises it further. The corners are dominated by interior
renewable-claim bundles: the correct normative object is not ``who owns''
but the stake profile $(w_{S}^{\star},w_{P}^{\star})$.
\end{enumerate}
\end{proposition}

\begin{IEEEproof}
(i) Substitute into \eqref{eq:wS}--\eqref{eq:wP} and apply
Lemma~\ref{lem:refresh}'s corner branch. (ii) Envelope on the owner's effort
choice; the FOC reduces to $w_{j}'(n)=0$, giving
Proposition~\ref{prop:scope}'s scopes. (iii)
$W^{E}=\Omega\,k_{j}w_{j}-\tfrac{k_{j}w_{j}^{2}}{2}$ at each corner's scope;
direct substitution. (iv) Strict concavity of $W(w_{S})$ in
Corollary~\ref{cor:stake}. (With the capitalization factor restored, the
steady-state flow versions are $W_{S}\propto k_{S}(3-\delta)(1-\lambda)^{4}$
and $W_{P}\propto k_{P}(3-\delta-4\lambda)$; the $\delta=1$ case reduces to
the displayed threshold, and the comparative statics are unchanged.)
\end{IEEEproof}

\begin{remark}[Compensated platform control; the earlier draft's
Proposition~2 recovered]\label{rem:compensated}
A third bundle completes the map: $g_{A}=(c{=}P,\alpha{=}0,\tau{=}\lambda)$,
platform control with full per-quality-deployment compensation, the
formalization of ``impeccable consent.'' Because compensation makes the
platform internalize $\lambda$, it chooses \emph{the same scope as the ought
corner}, $n_{A}=\tfrac{1-\lambda}{2}=\tfrac{n_{W}}{2}$, exactly the
earlier draft's Proposition~2 (``deployment exactly half the first best under
full compensation''), now derived inside the unified model. The opposite
margin, however, survives: $w_{P}=\tfrac{(1-\lambda)^{2}}{4}$, $w_{S}=0$.
Corner welfares become perfectly symmetric,
$W_{A}=\tfrac{k_{P}(1-\lambda)^{4}}{16}$ versus
$W_{S}=\tfrac{k_{S}(1-\lambda)^{4}}{16}$: once the deployment wedge is
priced, the choice between subject ownership and fully-compensating platform
ownership is \emph{purely} the question of whose refresh input is more
productive. One caution: $W_{A}<W_{P}$ iff $(1-\lambda)^{4}<1-2\lambda$,
i.e.\ for all $\lambda$ below $\lambda^{\dagger}\approx0.456$: in this
distortion-free benchmark, mandating full compensation at low exposure
\emph{lowers} welfare, because it taxes the only live refresh margin. The
efficiency case for compensation is not this static one; it is Section~\ref{sec:endog}'s
pollution channel, where compensation buys authenticity from strategic
subjects.
\end{remark}

\begin{remark}[Why the \emph{is} prevails, and what follows]\label{rem:defacto}
The is corner is not an accident of law; it is selected by feasibility.
Subject control presupposes the objective conditions the model names:
\emph{Separability} (Definition~\ref{def:artifact}), meaning the artifact
must be extractable from the host at all; functional portability $\phi$
(Definition~\ref{def:bundle}), meaning that exporting it must preserve
behavioral value; and attention $\zeta$, meaning the subject must perceive
her exposure. When $\phi\approx0$ or $\zeta\ll1$, $g_{S}$ is technologically
or behaviorally infeasible, and the market defaults to $g_{P}$ regardless of
what anyone thinks \emph{ought} to hold. Section~\ref{sec:endog} sharpens this
from a constraint into
an equilibrium: with naive subjects present, the platform profits from
opacity and zero compensation (obfuscation substitutes for compensation), so
the conditions that keep $g_{S}$ infeasible are themselves maintained by the
incumbent's optimal play: the is perpetuates itself. Two policy
consequences follow. First, by Theorem~\ref{thm:rights}(ii), \emph{renaming
the owner without moving $\phi$ and $\zeta$ changes neither scope nor
refresh}: nominal ownership reform is not the binding margin, portability and
salience are. Second, by (iii)--(iv), when those conditions can be moved, the
target is not the ought corner but the interior renewable-claim bundle,
since the ought, taken literally, would destroy the platform's refresh
margin exactly as the is destroys the subject's.
\end{remark}

\subsection{Taking stock}\label{sec:taking-stock}

The section has delivered three results, one mechanism, and a resolution of
the ought--is question.
Proposition~\ref{prop:ceiling} fixes the ultimate scarcity: decision value is
bounded by lived experience. Proposition~\ref{prop:scope} and
Corollary~\ref{cor:reversal} show that with the stock frozen, person linkage
already reverses the direction of the platform's distortion, and silently
destroys the subject's supply margin. Theorem~\ref{thm:rights} shows the
reversal cannot be repaired by reassigning ownership: co-production plus
budget balance leaves at least one margin under-incentivized, and control and
cash flow are non-substitutable instruments. Theorem~\ref{thm:headline}
delivers the dynamic consequence: because refresh is relationally supplied,
reach-expanding policy can trade current access against the long-run stock,
and Corollary~\ref{cor:access} shows the trade can be welfare-negative even
when each marginal deployment is statically valuable.
Proposition~\ref{prop:corners} then closes the framing: the ought and the is
are twin corners that each destroy one supply margin; their ranking is an
empirical question in $(\lambda,k_{S}/k_{P})$; and both are dominated by
interior bundles that keep \emph{both} co-producers on renewable claims.
Everything above holds under perfect provenance and complete information:
the failures are in the asset's person linkage and its co-produced renewal,
not in verification or beliefs. Section~\ref{sec:endog} endogenizes what was held fixed
here: the subject's option to distort her revealed behavior (pollution), the
participation margin that determines thickness (formation), and with them the
mechanism by which the is corner reproduces itself.

\section{Endogenous Quality and Market Formation}\label{sec:endog}

Section~\ref{sec:results} held two things fixed that no real memory market
holds fixed: the \emph{authenticity} of the subject's revealed behavior, and
the \emph{set} of subjects who supply at all. This section endogenizes both,
and in doing so completes the positive theory promised in
Section~\ref{sec:corners}: the \emph{is} corner is not merely inefficient;
it is \emph{self-perpetuating}, because the platform's equilibrium play
maintains the very conditions that make subject control infeasible.
Throughout, Assumptions~\ref{as:provenance}--\ref{as:completeinfo} remain in
force: everything below happens under perfect provenance, and the only
informational friction we add, the naive share $\nu$, is about
\emph{awareness of extraction}, not about fabricated lineage or hidden
quality.

\subsection{Setup: distortion, awareness, and compensation coverage}\label{sec:pollution-setup}

Before extraction, an \emph{aware} subject can distort her revealed behavior
at strictly convex cost $\kappa_{d}(d)$, $\kappa_{d}(0)=\kappa_{d}'(0)=0$,
$\kappa_{d}''>0$, reducing effective fidelity to $\tilde q=(1-d)Q$ (closed
forms use $\kappa_{d}(d)=c_{d}d^{2}/2$). Distortion is the individually
rational counterpart of Proposition~\ref{prop:ceiling}: the subject cannot
raise the ceiling of what her history contains, but she can lower what
extraction recovers from it.

Two objects are derived from Section~\ref{sec:model} rather than introduced
by assumption. The subject's \emph{compensation coverage} per unit of
quality-deployment is
\begin{equation}\label{eq:coverage}
s:=\tau+\frac{\alpha R(n)}{n},
\end{equation}
her total receipts per quality-deployment under the rights bundle $g$; her
net exposure per unit of effective quality is then $(\lambda-s)D$, where $D$
is the \emph{deployment scale}. In the single-unit model of
Section~\ref{sec:model}, $D=n$; we keep it as a separate symbol because the
aggregate market of Section~\ref{sec:formation-setup} multiplies it by
thickness. A share $\nu\in[0,1)$ of subjects is \emph{naive} (unaware that
extraction occurs); the platform chooses opacity $o\in\{0,1\}$ ($o=1$
preserves naivety; $o=0$, or a transparency mandate $\mathsf T=1$, converts
naive into aware).
Naivety is the discrete counterpart of the attention parameter $\zeta$ of
Section~\ref{sec:attention}.

\subsection{Pollution and instrument complementarity}\label{sec:pollution}

\begin{theorem}[Pollution and complementarity]\label{thm:pollution}
\begin{enumerate}
\item[(i)] \textbf{(Endogenous fidelity ceiling.)} The aware subject's
distortion satisfies
\begin{equation}\label{eq:distortion-foc}
\kappa_{d}'(d^{*})=(\lambda-s)^{+}\,QD,
\end{equation}
so $d^{*}>0$ iff exposure is uncompensated ($s<\lambda$), and in that case
effective fidelity $\tilde q^{*}=(1-d^{*})Q$ is strictly decreasing in market
scale $D$: the platform-controlled memory market pollutes its own raw
material as it grows. In the quadratic benchmark
$d^{*}=\min\{1,(\lambda-s)nQ/c_{d}\}$.
\item[(ii)] \textbf{(Compensation is quality control.)} If $s\ge\lambda$,
then $d^{*}=0$ and, since $s\ge\lambda$ implies
$w_{S}=u_{S}+\alpha R+(\tau-\lambda)n\ge u_{S}\ge0$ with strict inequality
whenever coverage is strict, refresh effort recovers along with authenticity.
Paying the subject is not (only) redistribution; it purchases authenticity.
Asset quality is a function of the institution.
\item[(iii)] \textbf{(Obfuscation equilibrium and instrument
complementarity.)} Naive subjects neither distort nor chill and are the
platform's cheapest high-quality supply, so for $\nu>0$ the platform's
one-shot optimum is $(o,s)=(1,0)$: obfuscation substitutes for compensation.
Consequently: a transparency mandate alone ($\mathsf T=1$, $s=0$) converts
naive into distorting subjects, so effective quality and buyer surplus fall;
compensation alone ($\mathsf T=0$, $s\ge\lambda$) restores the aware
subjects' authenticity but leaves the naive exposed without informed consent;
only the pair ($\mathsf T=1$ and $s\ge\lambda$), or subject control itself,
restores authenticity and voluntariness together. Transparency and
compensation are complements, not substitutes.
\item[(iv)] \textbf{(Commitment.)} A promised $s\ge\lambda$ is not credible
in the one-shot game: once $\tilde q$ is supplied, compensation is a pure
transfer and the platform reneges; anticipating this, aware subjects distort.
A rule-based cession of revenue rights (tokenized royalties or governance
commitments in the sense of \cite{sockinxiong2023}) implements
$s\ge\lambda$ and is in the platform's own interest iff
\begin{equation}\label{eq:commit}
(1-\nu)\,R(n)\,Q>c_{d}:
\end{equation}
few naive suppliers, cheap distortion technology, high stakes. Where the
inequality fails (abundant naivety), the platform strictly prefers
opacity, which is part (iii)'s equilibrium.
\end{enumerate}
\end{theorem}

\begin{IEEEproof}
(i) The aware subject minimizes
$(\lambda-s)^{+}(1-d)QD+\kappa_{d}(d)$; the FOC is
\eqref{eq:distortion-foc}, and $d^{*}>0$ iff the marginal benefit at $d=0$ is
positive, i.e.\ $s<\lambda$;
$\partial d^{*}/\partial D=(\lambda-s)Q/\kappa_{d}''>0$ by the implicit
function theorem, so $\partial\tilde q^{*}/\partial D<0$. (ii) At
$s\ge\lambda$ the objective is nonincreasing in $d$; $d^{*}=0$; the stake
bound is \eqref{eq:wS}. (iii) Platform flow profit per subject is
$[R(n)-sn]\tilde q(s;\text{awareness})$; a naive subject supplies
$\tilde q=Q$ at any $s$, so profit from the naive is maximized at $s=0$, and
$o=1$ preserves that supply at zero cost; an aware subject at $s=0$ supplies
$Q(1-d^{*}(0))$. Hence
$\Pi(o{=}1,s{=}0)-\Pi(o{=}0,s{=}0)=\nu R(n)Q\,d^{*}(0)>0$ for $\nu>0$,
$d^{*}(0)>0$: opacity strictly dominates voluntary transparency. The welfare
statements follow cell by cell: $(\mathsf T{=}1,s{=}0)$ replaces naive supply
$Q$ by $Q(1-d^{*}(0))$ (quality and buyer surplus fall by $R$- and
$C$-proportionality); $(\mathsf T{=}0,s{\ge}\lambda)$ leaves a measure $\nu$
of subjects bearing $\lambda\tilde qD$ without informed choice: monetary
compensation reaches them, but voluntariness (the constraint
$\ell\in\mathcal R_{i}(m)$, Assumption~\ref{as:commitment}) does not; both
instruments together set $d^{*}=0$ with all subjects aware. (iv) The renege
gain is $sn\tilde q>0$, so any $s>0$ is sequentially irrational without a
binding rule. Under the rule, per-subject profit is $(R(n)-\lambda n)Q$
against $R(n)Q[\nu+(1-\nu)(1-d^{*}(0))]$ under opacity; in the quadratic
benchmark the difference is
$\lambda nQ\,[(1-\nu)R(n)Q-c_{d}]/c_{d}$, positive iff \eqref{eq:commit}.
\end{IEEEproof}

\begin{remark}[Misrepresentation robustness]\label{rem:misrep}
With exposure $\lambda_{P}\tilde q+\lambda_{M}(1-\tilde q)$, distortion's
private return falls; for large $\lambda_{M}$ the subject prefers
high-fidelity-or-exit, and the policy content of (iii) extends from
compensation to \emph{genuinely feasible exit rights}.
\end{remark}

\begin{corollary}[The \emph{is} corner is self-perpetuating]\label{cor:selfperp}
In the obfuscation equilibrium of Theorem~\ref{thm:pollution}(iii), the
platform's optimal play (a) maintains $o=1$, keeping the aware share (and
hence effective attention $\zeta$) low, and (b) supplies no functional
portability, since $\phi>0$ only improves the subject's outside option and
the feasibility of $g_{S}$ without raising platform revenue. The objective
conditions for subject control identified in Section~\ref{sec:corners}
(Separability exercised through $\phi$, and awareness $\zeta$) are therefore
endogenously unmet on the equilibrium path: absent intervention that targets
$\phi$ and $\zeta$ directly, the de facto corner $g_{P}$ reproduces itself.
By Theorem~\ref{thm:rights}(ii), intervention that renames the owner while
leaving $\phi$ and $\zeta$ in place does not disturb this equilibrium.
\end{corollary}

\begin{IEEEproof}
(a) is Theorem~\ref{thm:pollution}(iii). (b): platform revenue
$[R(n)-sn]\tilde q$ is independent of $\phi$, while $\phi>0$ weakly raises
the subject's outside option $\bar U_{S}$ and hence the $F$ needed for IR:
portability is weakly costly, strictly so when exit is an equilibrium threat.
The last sentence is Theorem~\ref{thm:rights}(ii).
\end{IEEEproof}

\subsection{Setup: participation and thickness}\label{sec:formation-setup}

Subjects are heterogeneous in sensitivity $\lambda_{i}\sim G$ (smooth,
positive density $g$). Subjects with $\lambda_{i}$ below a threshold list;
the listing mass is the market thickness $M=G(\hat\lambda)$ of
Definition~\ref{def:stock}, and the match value of a listing is
\begin{equation}\label{eq:theta}
\Theta(M):=\tilde\Theta(\Psi(M)),\qquad\tilde\Theta'>0,
\end{equation}
increasing in the market precision
$\Psi(M)=(\sigma_{0}^{2}+1/(M\bar\psi))^{-1}$: thicker markets aggregate the
common component better, and buyers pay for that precision. $\Theta$ inherits
Definition~\ref{def:stock}'s curvature: increasing, concave, and
\emph{bounded}, with $\Theta'(M)\to0$ as $M$ grows. Aggregation saturates
because $\sigma_{0}^{2}>0$ is noise no thickness can remove. Let $v>0$ be the
match-surplus fundamental per listing, of which the listing subject
appropriates the coverage share $s$ (Section~\ref{sec:pollution-setup}).

Each listing also leaks: by Definition~\ref{def:state}, within-class personal
components are correlated ($\varrho>0$), so a listed artifact is informative
about \emph{non-listing} peers; define the \emph{relational leakage
intensity} $\Lambda(\varrho)\ge0$, $\Lambda(0)=0$, $\Lambda'>0$: a listing
imposes expected exposure $\Lambda(\varrho)\,\bar\lambda^{\mathrm{nl}}\,
\tilde qD$ on correlated non-listers, where $\bar\lambda^{\mathrm{nl}}$ is
their mean sensitivity. Private listing decisions ignore this term entirely;
it bypasses even perfect consent, because the person harmed never signed
anything.

\subsection{Formation, the three wedges, and over-participation}\label{sec:formation}

\begin{theorem}[Formation and the direction of the participation
distortion]\label{thm:formation}
\begin{enumerate}
\item[(i)] \textbf{(Threshold equilibria and selection.)} Equilibria are
threshold rules: subjects with $\lambda_{i}\le\hat\lambda^{*}$ list, where
$\hat\lambda^{*}$ solves
\begin{equation}\label{eq:market-threshold}
\hat\lambda=s\,v\,\Theta(G(\hat\lambda)).
\end{equation}
Thickness--participation complementarity ($\Theta'>0$) can generate multiple
equilibria: a self-fulfilling thin trap coexisting with a thick market.
With private noise about the fundamental $v$, a global-games perturbation
selects a unique threshold and restores unambiguous comparative statics
(standard transplant of \cite{morrisshin1998}).
\item[(ii)] \textbf{(Three wedges; the direction is an \emph{iff}.)} The
planner's threshold satisfies
\begin{equation}\label{eq:planner-threshold}
\hat\lambda^{W}=v\,\Theta(G)+v\,\Theta'(G)\,G
-\Lambda(\varrho)\,P(\hat\lambda^{W}),
\end{equation}
where $P(\hat\lambda):=\bar\lambda^{\mathrm{nl}}(\hat\lambda)\,\tilde qD
\times$ (mass of affected peers per listing) is the marginal relational
harm. Comparing with \eqref{eq:market-threshold}, the market threshold
exceeds the planner's (over-participation) iff
\begin{equation}\label{eq:wedges}
\Lambda(\varrho)\,P\;>\;
\underbrace{(1-s)\,v\,\Theta}_{\text{appropriability wedge}}
\;+\;\underbrace{v\,\Theta'(G)\,G}_{\text{thickness externality}}.
\end{equation}
The marginal lister ignores one negative externality (relational leakage)
and two positive margins (unappropriated match surplus and her contribution
to everyone else's $\Theta$); which force wins is a measurable comparison,
not an axiom.
\item[(iii)] \textbf{(Why person-linked memory markets tend to be too
thick.)} As the market thickens, the right side of \eqref{eq:wedges}
attenuates and the left side does not: $\Theta'(M)\to0$ by
Definition~\ref{def:stock}'s saturation, while $\Lambda(\varrho)P$ is
anchored by the class correlation $\varrho$, which for preference/policy
memory is largest exactly among similar people, the defining feature of
the focal class (Definition~\ref{def:state}). Hence for any $\varrho>0$ and
coverage $s$ close to full appropriation, there is a finite thickness beyond
which \eqref{eq:wedges} holds and every additional lister is socially
excessive. Naive participation ($\nu>0$, or $\zeta<1$) shifts the market
threshold further right at unchanged welfare, amplifying over-participation
independently. In contrast to thin-market underprovision in the received
data-economics literature, identity-entangled memory markets over-thicken at
scale.
\item[(iv)] \textbf{(Non-monotone welfare.)} Welfare
$W(\hat\lambda)=\int_{0}^{\hat\lambda}[v\Theta(G(\hat\lambda))-\ell]\,
dG(\ell)-\Lambda(\varrho)\,G(\hat\lambda)\,
\bar\lambda^{\mathrm{nl}}(\hat\lambda)\tilde qD$ is non-monotone in the
threshold: it rises while low-$\lambda$ subjects enter and falls once
\eqref{eq:wedges} binds. There are parameter regions ($\Lambda(\varrho)$
large, mass of $G$ at high $\lambda$, or large naive share) where marginal
growth subsidies destroy welfare and where the voluntary formation of the
market is itself welfare-negative: relational externalities bypass
participants' revealed preference, and naive participation is not informed
revelation.
\end{enumerate}
\end{theorem}

\begin{IEEEproof}[Proof sketch]
(i) Best responses are monotone in the conjectured threshold; existence by
Tarski; multiplicity when $svG'\Theta'>1$ somewhere; global-games uniqueness
by strategic complementarity plus dominance regions at $v$ extremes.
(ii) Differentiate $W$ at the threshold:
$W'(\hat\lambda)=g(\hat\lambda)[v\Theta-\hat\lambda+v\Theta'G-\Lambda P]$;
the market condition replaces the bracket by $sv\Theta-\hat\lambda$;
subtract. (iii)
$\Psi'(M)=\bar\psi/(M\bar\psi\sigma_{0}^{2}+1)^{2}\to0$; $\Lambda(\varrho)P$
bounded away from zero for $\varrho>0$ with nondegenerate
$\bar\lambda^{\mathrm{nl}}$; naive listers use $\zeta\lambda_{i}<\lambda_{i}$
in the private cutoff while $W$ evaluates true $\lambda_{i}$. (iv) Sign
change of $W'$ from (ii)--(iii); the formation statement compares
$W(\hat\lambda^{*})$ with $W(0)=0$. Appendix~\ref{app:formation} gives
the full argument, including the single-crossing condition under which the
two thresholds can be ordered.
\end{IEEEproof}

\begin{remark}[Correction relative to an earlier draft]\label{rem:correction}
An earlier statement claimed over-participation unconditionally whenever
relational leakage is positive. That is false as stated: the marginal lister
also confers the two unpriced \emph{positive} margins in \eqref{eq:wedges},
and at low thickness they can dominate: a thin memory market can genuinely
be too thin. The defensible claim is the \emph{iff} in (ii) plus the
saturation argument in (iii): the presumption of over-participation is an
asymptotic property of thick, high-$\varrho$, well-appropriated markets, and
\eqref{eq:wedges} tells the econometrician exactly which three objects to
measure. Regulators should track relational-externality density, not listing
counts (cf.\ Section~\ref{sec:governance}).
\end{remark}

\subsection{Selection into authorization}\label{sec:selection}

Formation determines \emph{who lists}; institutions determine \emph{what
listing signals}. Fix a class of artifacts with heterogeneous effective
fidelity $\tilde q$ and subjects private in $\lambda_{i}\sim G$, independent
of $\tilde q$; a marketplace pays a royalty per deployment.

\begin{proposition}[An authorization badge is not a fidelity
certificate]\label{prop:selection}
\begin{enumerate}
\item[(i)] Under a \emph{flat} royalty $r$ per deployment, subject $i$
authorizes iff $\lambda_{i}\tilde qD\le r$, so
$\Pr(\text{authorized}\mid\tilde q)=G(r/(\tilde qD))$, strictly decreasing
in $\tilde q$: the authorized pool is first-order stochastically worse in
fidelity, so the badge is a negative quality signal.
\item[(ii)] If misrepresentation harm dominates (exposure decreasing in
$\tilde q$), the direction reverses; authorization is generically
informative about fidelity in a direction set by the harm technology, never
a neutral stamp.
\item[(iii)] Under the quality-indexed royalty of
Definition~\ref{def:bundle} (compensation $\tau$ per
\emph{quality}-deployment, i.e.\ $r(\tilde q)=\tau\tilde qD$), the
authorization condition becomes $\lambda_{i}\le\tau$: participation is
independent of $\tilde q$ and selection is shut down.
\end{enumerate}
\end{proposition}

\begin{IEEEproof}
(i) The cutoff $r/(\tilde qD)$ falls in $\tilde q$; FOSD is monotonicity of
$G$. (ii) Replace $\lambda_{i}\tilde q$ by $\lambda_{i}h(\tilde q)$ with
$h'<0$ and repeat. (iii) $\tilde qD$ cancels.
\end{IEEEproof}

\begin{remark}\label{rem:why-tau}
Proposition~\ref{prop:selection}(iii) is the reason
Definition~\ref{def:bundle} defines $\tau$ per quality-deployment rather
than per deployment: the indexation is not a modeling convenience but the
instrument that makes authorization \emph{uninformative about quality},
which is what allows permission credentials and performance credentials to
be issued separately (Section~\ref{sec:governance}). The economically
correct role of an independent certifier is therefore not to substitute for
consent but to make $\tilde q$ contractible, enabling exactly this
indexation. Under flat royalties, empirical fidelity comparisons between
authorized and unauthorized pools are confounded by (i);
Section~\ref{sec:empirics} turns this into a test.
\end{remark}

\subsection{The reach bridge: thickness feeds depletion}\label{sec:bridge}

\begin{corollary}[Thickness transmits into the stock]\label{cor:bridge}
With the reach bridge $n=N(M)$, $N'\ge0$ (Definition~\ref{def:stock}), the
long-run stock responds to thickness as
\begin{equation}\label{eq:bridge}
\frac{d\bar Q}{dM}=\frac{1}{\delta}\,\frac{dY^{E}}{dn}\,N'(M),
\end{equation}
with $dY^{E}/dn$ from Theorem~\ref{thm:headline}; and effective delivered
quality is $\tilde q(s)\times\bar Q$, with $\tilde q$ from
Theorem~\ref{thm:pollution}. Hence market growth acts on the memory stock
through three multiplicative channels: participation
(Theorem~\ref{thm:formation}), reach (Theorem~\ref{thm:headline}), and
authenticity (Theorem~\ref{thm:pollution}). It can therefore raise measured
access while lowering both the stock and its fidelity when coverage $s$ is
low.
Full tipping dynamics and the relational-data-externality fixed point are
developed in the aggregate extension.
\end{corollary}

\subsection{Limited attention}\label{sec:attention}

If the subject perceives only a share $\zeta\in[0,1]$ of true exposure, her
perceived stake replaces $\lambda$ by $\zeta\lambda$; effort \emph{rises} as
attention falls ($\partial e^{E}/\partial\zeta=-\eta_{S}\lambda n/c_{S}<0$),
distortion falls, and observed consent weakly expands. But welfare must be
evaluated at the true $\lambda$: higher measured participation and higher
measured quality under low attention are not evidence of higher welfare;
they are Theorem~\ref{thm:pollution}(iii)'s obfuscation equilibrium in
continuous form. $\zeta=1$ recovers full rationality, and the
sophistication-selection conjecture (Section~\ref{sec:limits}) is the
cross-sectional implication: the market sources its best raw material from
those least able to protect themselves.

\subsection{Where this leaves the institution space}\label{sec:endog-stock}

Theorems~\ref{thm:pollution}--\ref{thm:formation} add the last two rows to
the paper's instrument logic. Compensation coverage $s$ is \emph{quality
policy} (Theorem~\ref{thm:pollution}(ii)) and, indexed per
quality-deployment, \emph{selection policy}
(Proposition~\ref{prop:selection}(iii)); transparency $\mathsf T$ is
\emph{voluntariness policy}, worthless for quality on its own
(Theorem~\ref{thm:pollution}(iii)); portability $\phi$ and attention $\zeta$
are \emph{feasibility policy} for subject control, and no private actor
supplies them in equilibrium (Corollary~\ref{cor:selfperp}); participation
volume is not a health metric, and the object to monitor is
relational-externality density \eqref{eq:wedges}. None of these instruments
substitutes for another, and none is delivered by provenance, which has
been perfect on every page of this section.

\section{The Optimal Rights Bundle}\label{sec:optimal}

Sections~\ref{sec:results}--\ref{sec:endog} were diagnostic: the ownership
corners each destroy a supply margin (Proposition~\ref{prop:corners}),
uncompensated markets pollute their raw material
(Theorem~\ref{thm:pollution}), and formation can over-thicken
(Theorem~\ref{thm:formation}). This section is constructive, solving for
the \emph{second-best rights bundle}: the choice of $(\alpha,\tau,n)$
that maximizes welfare subject to the frictions the model takes as
primitive, namely noncontractible efforts, budget balance, and the
subject's option to distort. Two regimes are analyzed: \emph{commitment}
(a planner, or a marketplace constitution, fixes scope and compensation
directly) and \emph{delegation} (the planner fixes only the cash-flow
bundle $(\alpha,\tau)$; the controller named in $c$ chooses scope).
Throughout we work in the additive benchmark with $u_{S}=u_{P}=0$; general
flows are restored in Appendix~\ref{app:flows}.

\subsection{Coverage is the single static instrument}\label{sec:collapse}

The bundle $(\alpha,\tau)$ appears to give the designer two compensation
instruments, but it does not.

\begin{proposition}[Instrument collapse and the financing slope]\label{prop:collapse}
\begin{enumerate}
\item[(i)] At any scope $n$, the stakes \eqref{eq:wS}--\eqref{eq:wP} depend
on $(\alpha,\tau)$ only through the coverage $s=\tau+\alpha R(n)/n$ of
Section~\ref{sec:pollution-setup}:
\begin{equation}\label{eq:collapse}
w_{S}=u_{S}+(s-\lambda)\,n,\qquad w_{P}=u_{P}+R(n)-s\,n.
\end{equation}
Revenue shares and per-quality royalties are perfect substitutes for
refresh incentives at fixed scope.
\item[(ii)] They are not substitutes across scopes: holding $(\alpha,\tau)$
fixed, $\partial s/\partial n=-\alpha$. Royalty-financed coverage is
reach-invariant; share-financed coverage dilutes as reach expands, and
under pure share financing ($\tau=0$) coverage crosses the authenticity
floor $\lambda$ exactly at $n=2n_{S}$, twice the subject's own preferred
scope (Proposition~\ref{prop:scope}). Expansion under share financing
therefore drives the market into Theorem~\ref{thm:pollution}'s pollution
region mechanically, with no change in the written contract.
\item[(iii)] Both instruments are quality-indexed in the sense of
Proposition~\ref{prop:selection}(iii): neither induces selection on
fidelity. The designer's static choice is one number, $s$; the
$\alpha$--$\tau$ split matters only through (ii) and through delegation
(Section~\ref{sec:delegation}).
\end{enumerate}
\end{proposition}

\begin{IEEEproof}
(i) $\alpha R+\tau n=n\,s$ by definition of $s$; substitute into
\eqref{eq:wS}--\eqref{eq:wP}. (ii) $s(n)=\tau+\alpha(1-n)$; solve
$s(n)=\lambda$ and compare with Proposition~\ref{prop:scope}'s $n_{S}$.
(iii) Both scale receipts with $\tilde q$; apply
Proposition~\ref{prop:selection}(iii).
\end{IEEEproof}

\subsection{The second-best bundle under commitment}\label{sec:secondbest}

The commitment program is
\begin{equation}\label{eq:program}
\max_{n\in[0,1],\,s\ge0}\;W^{E}(n,s)=\Omega(n)\,Y^{E}
-\tfrac{k_{S}w_{S}^{2}}{2}-\tfrac{k_{P}w_{P}^{2}}{2}
\end{equation}
subject to Lemma~\ref{lem:refresh} and $w_{P}\ge0$, on the no-pollution
region $s\ge\lambda$ first; Section~\ref{sec:regimes} opens the polluted
region.

\begin{theorem}[The second-best bundle]\label{thm:secondbest}
\begin{enumerate}
\item[(i)] At any scope, optimal coverage implements
Corollary~\ref{cor:stake}'s stake profile:
\begin{equation}\label{eq:sstar}
s^{*}(n)=\lambda+\frac{w_{S}^{\star}(n)}{n},
\end{equation}
so the subject's compensation exceeds bare exposure by exactly her optimal
renewable claim per unit of reach.
\item[(ii)] \textbf{(The authenticity floor is self-enforcing.)}
$s^{*}(n)\ge\lambda$ iff $w_{S}^{\star}(n)\ge0$ iff
\begin{equation}\label{eq:floor}
\frac{k_{S}}{k_{P}}\;\ge\;\frac{C(n)}{\Omega(n)}=\frac{n}{2-2\lambda-n}:
\end{equation}
whenever the subject's relative refresh productivity exceeds the
buyer-surplus share of social value, the second best pays above exposure
without any authenticity mandate; Theorem~\ref{thm:pollution}(ii)'s
floor is implied, not imposed.
\item[(iii)] \textbf{(Closed form, symmetric productivities.)} For
$k_{S}=k_{P}=k$, the concentrated program
$W^{*}(n)=k\,n^{2}(1-\lambda-n)\bigl(3(1-\lambda)-n\bigr)/4$ has the unique
interior maximizer
\begin{equation}\label{eq:nstar}
\begin{split}
&n^{*}=\tfrac{3-\sqrt3}{2}\,(1-\lambda)\approx0.634\,(1-\lambda),\\
&s^{*}=\lambda+\tfrac{\sqrt3-1}{4}\,(1-\lambda),
\end{split}
\end{equation}
strictly between the subject-corner scope $\tfrac{1-\lambda}{2}$ and the
static welfare scope $1-\lambda$.
\item[(iv)] \textbf{(The one-third depletion haircut.)} In the polar
single-margin cases (the projection binds: $w_{S}^{\star}\in\{0,T\}$), the
concentrated program yields $n^{*}=\tfrac{2}{3}(1-\lambda)$. Hence across
all productivity ratios,
\begin{equation}\label{eq:band}
n^{*}\in\Bigl[\tfrac{3-\sqrt3}{2},\,\tfrac{2}{3}\Bigr]\times(1-\lambda)
\approx[0.63,\,0.67]\times(1-\lambda),
\end{equation}
decreasing in $\lambda$: the dynamic second best deploys roughly two-thirds
of the static first-best reach, almost independently of who the productive
refresher is. The missing third is the price of keeping refresh incentives
alive: the depletion externality of Theorem~\ref{thm:headline}, now
internalized as a scope haircut.
\end{enumerate}
\end{theorem}

\begin{IEEEproof}
(i)--(ii): $\partial W^{E}/\partial s
=n[(k_{S}-k_{P})\Omega-k_{S}w_{S}+k_{P}w_{P}]$; the root reproduces
Corollary~\ref{cor:stake}'s $w_{S}^{\star}$, and
$s^{*}-\lambda=w_{S}^{\star}/n$; the productivity condition is
$w_{S}^{\star}\ge0$ rearranged with $\Omega=T+C$. (iii): substitute
$s^{*}(n)$, factor, take the interior root of the FOC; the second candidate
exceeds $1-\lambda$ and is infeasible. (iv): with a cornered stake the
program is $\max_{n}\,k_{j}[\Omega T-\tfrac{T^{2}}{2}]$, whose FOC has the
interior root $\tfrac{2}{3}(1-\lambda)$. Appendix~\ref{app:secondbest}
shows that the concentrated program is unimodal on $(0,1-\lambda)$ for every
productivity ratio, gives the closed form
\eqref{eq:nstar-general} for $n^{*}$ at an arbitrary ratio, and establishes
that the band \eqref{eq:band} is exact and attained, so no quasiconcavity
regularity condition is needed.
\end{IEEEproof}

\begin{remark}\label{rem:two-inputs}
Theorem~\ref{thm:secondbest} converts Section~\ref{sec:corners}'s ought--is
stalemate into a design rule with two estimable inputs. The scope rule
barely needs the productivity split ($n^{*}$ moves within a four-point
band); the compensation rule needs it exactly ($s^{*}$ is linear in
$w_{S}^{\star}$, which is linear in $k_{S}-k_{P}$). An econometrician who
can estimate only $\lambda$ can already set scope nearly optimally; setting
compensation requires Section~\ref{sec:empirics}'s decomposition of refresh
into subject and platform contributions.
\end{remark}

\subsection{When is pollution tolerated?}\label{sec:regimes}

Dropping the constraint $s\ge\lambda$ and letting effective quality carry
Theorem~\ref{thm:pollution}'s multiplier $\tilde q(s)=1-d^{*}(s)$
(quadratic benchmark, $d^{*}=(\lambda-s)^{+}nQ/c_{d}$), the objective
acquires a kink at $s=\lambda$. Write $a_{1}:=1-\lambda-n>0$ and
$a_{2}:=2-2\lambda-n>0$ on the active region.

\begin{proposition}[Three coverage regimes]\label{prop:regimes}
The second-best coverage satisfies:
\begin{enumerate}
\item[(a)] \textbf{Strict renewable claim}, $s^{*}>\lambda$, iff
$k_{S}/k_{P}>n/a_{2}$ (Theorem~\ref{thm:secondbest}(ii));
\item[(b)] \textbf{Floor exactly}, $s^{*}=\lambda$: the subject is fully
covered but holds no net claim; no pollution and zero distortion, which is
the corner regime between (a) and (c);
\item[(c)] \textbf{Tolerated pollution}, $s^{*}<\lambda$, iff
\begin{equation}\label{eq:tolerate}
c_{d}\bigl[(k_{P}-k_{S})\,a_{2}-2k_{P}\,a_{1}\bigr]\;>\;
Q\,k_{P}\,n\,a_{1}a_{2},
\end{equation}
i.e., only when the platform is the strongly dominant refresher
\emph{and} the distortion technology is stiff ($c_{d}$ large relative to
the stakes $QR$): the subject is simultaneously unproductive and
defenseless, so buying her authenticity is not worth the platform stake it
costs.
\end{enumerate}
\end{proposition}

\begin{IEEEproof}
One-sided derivatives of the kinked objective at $s=\lambda$: from above,
$\operatorname{sign}=\operatorname{sign}(w_{S}^{\star})$; from below, the
displayed expression, computed in closed form in
Appendix~\ref{app:regimes}.
\end{IEEEproof}

\begin{remark}[Normative caveat]\label{rem:normative}
Regime (c) is a transferable-utility statement and nothing more: it prices
the subject's distortion, not her autonomy. Under
Assumption~\ref{as:commitment} the inalienable-use constraint
$\ell\in\mathcal R_{i}(m)$ can exclude (c) outright, and under the
misrepresentation variant (Remark~\ref{rem:misrep}) the regime shrinks
further, since low fidelity then harms the subject directly. We report (c)
because suppressing it would hide what pure efficiency logic recommends,
and that is precisely the kind of claim a welfare analysis should expose
to normative scrutiny rather than bury.
\end{remark}

\subsection{Delegated scope: who should control deployment}\label{sec:delegation}

Suppose the planner sets only $(\alpha,\tau)$ and the controller named in
$c$ chooses scope ex ante (Theorem~\ref{thm:rights}(ii)'s conditions). Can
the commitment optimum $(n^{*},s^{*})$ be reproduced? By
Proposition~\ref{prop:collapse}, delegation gives the planner two
instruments for two targets (the controller's scope FOC and the coverage
identity), so the question is feasibility, not counting.

\begin{theorem}[Delegated implementation is asymmetric]\label{thm:delegation}
\begin{enumerate}
\item[(i)] \textbf{(Platform control.)} The unique candidate pair is
\begin{equation}\label{eq:deleg-P}
\alpha^{*}=\frac{s^{*}-1+2n^{*}}{n^{*}},\qquad
\tau^{*}=(1-\alpha^{*})(1-2n^{*}),
\end{equation}
and it is feasible ($\tau^{*}\ge0$, $\alpha^{*}\in[0,1]$) only if
$n^{*}\le\tfrac12$. In the symmetric benchmark this holds iff
\begin{equation}\label{eq:lambdac}
\lambda\;\ge\;\lambda_{c}=\frac{3-\sqrt3}{6}\approx0.211.
\end{equation}
For low-exposure classes the platform under-deploys and per-quality
royalties can only shrink its reach further; implementing $n^{*}>\tfrac12$
requires $\tau^{*}<0$, that is, per-deployment subsidies (cf.
Proposition~\ref{prop:renewable}'s financing caveat).
\item[(ii)] \textbf{(Subject control.)} The unique candidate pair is
\begin{equation}\label{eq:deleg-S}
\alpha^{*}=\frac{w_{S}^{\star}}{(n^{*})^{2}},\qquad
\tau^{*}=\lambda-\alpha^{*}(1-2n^{*}),
\end{equation}
and in the symmetric benchmark it is feasible for every $\lambda\in(0,1)$,
with the remarkably simple form $\alpha^{*}=\tfrac{1}{2\sqrt3}\approx0.289$
(constant in $\lambda$) and $\tau^{*}>0$ increasing in $\lambda$.
\item[(iii)] \textbf{(Reading.)} The feasibility map answers
Section~\ref{sec:corners}'s control question constructively: the
\emph{ought} controller is delegation-feasible for every exposure class,
while the \emph{is} controller is delegation-feasible only where exposure
is high ($\lambda\ge\lambda_{c}$), which are exactly the classes where, by
Corollary~\ref{cor:reversal}, platform scope is least distorted to begin
with. Where the de facto regime cannot be steered to the second best with
nonnegative royalties, the options are direct scope regulation
(Section~\ref{sec:secondbest}), deployment subsidies, or reassigning
control; only the last also repairs the feasibility conditions of
Corollary~\ref{cor:selfperp}.
\end{enumerate}
\end{theorem}

\begin{IEEEproof}
Substitute the controller FOCs of Proposition~\ref{prop:scope} into the
coverage identity: under platform control $s=(1-2n)+\alpha n$ along the FOC
locus, under subject control $s=\lambda+\alpha n$; solve each pair for
$(\alpha,\tau)$. Feasibility bounds by direct evaluation; the symmetric
constants follow from $n^{*}=\tfrac{3-\sqrt3}{2}(1-\lambda)$ and
$w_{S}^{\star}=T/2$, whence
$\alpha^{*}=(1-\lambda-n^{*})/(2n^{*})=\tfrac{1}{2\sqrt3}$.
\end{IEEEproof}

\subsection{Taking stock: the constitution of a memory market}\label{sec:constitution}

Assembling Sections~\ref{sec:results}--\ref{sec:optimal}, the second-best
constitution of a personalized-memory market has five components, none of
which is ``an owner.'' \emph{Coverage} $s^{*}=\lambda+w_{S}^{\star}/n$:
compensation above exposure by the subject's optimal renewable claim,
royalty-financed rather than share-financed where reach is expected to grow
(Proposition~\ref{prop:collapse}(ii)), committed through rules rather than
promises (Theorem~\ref{thm:pollution}(iv)). \emph{Scope}
$n^{*}\approx\tfrac23(1-\lambda)$: the static welfare rule with a
one-third depletion haircut (Theorem~\ref{thm:secondbest}(iv)), either set
directly or delegated subject to Theorem~\ref{thm:delegation}'s feasibility
map. \emph{Transparency} $\mathsf T=1$, which is complementary to coverage,
not a substitute (Theorem~\ref{thm:pollution}(iii)). \emph{Feasibility
infrastructure}, portability $\phi$ and salience $\zeta$, supplied by
policy because no equilibrium actor supplies them
(Corollary~\ref{cor:selfperp}). And \emph{formation monitoring} keyed to
relational-externality density rather than listing counts
(Theorem~\ref{thm:formation}). The ought--is question dissolves on contact
with this list: what the subject should hold is not the asset but a claim
structure; what the platform should hold is not the asset but a residual
stake plus scope constrained by its own feasibility region; and what
neither should hold is the other's margin.

\subsection{Information and market-design extensions (compressed)}\label{sec:extensions}

Everything above holds under complete information
(Assumption~\ref{as:completeinfo}). Relaxing it adds a familiar layer,
which we compress because it is the least memory-specific part of the
problem and is developed in Appendix~\ref{app:info}: private fidelity generates an
Akerlof pool \cite{akerlof1970} whose repair is performance certification,
distinct from provenance certification and subject to certifier market
power in the sense of Lizzeri \cite{lizzeri1999}; pre-purchase evaluation
leaks the asset, pushing high-$\kappa^{\star}$ classes toward query-sandbox
access and audited replay \cite{antonyao2002}; and compatibility-indexed
evaluation is forced by Lemma~\ref{lem:no-score} regardless of information
structure. The certifier's economically correct role was already pinned
down with complete information (Remark~\ref{rem:why-tau}): make $\tilde q$
contractible so that coverage can be quality-indexed; certification
serves authorization and never substitutes for it. Multi-subject relational
memory remains companion-paper scope, with
Theorem~\ref{thm:formation}'s machinery as the vehicle.

\section{Governance and Design Implications}\label{sec:governance}

\subsection{Reading the model as a design document}\label{sec:design-doc}

Every claim in this section is a corollary of
Sections~\ref{sec:results}--\ref{sec:optimal}, cited by result; none is a
prior. This discipline matters because governance debates about AI memory
currently run on intuitions (``data wants to be shared,'' ``users should
own their data,'' ``markets will under-provide quality''), each of which
the model shows to be true only on a bounded parameter region, and each of
which, applied outside its region, prescribes the wrong instrument. Three
kinds of actors implement what follows: the \emph{marketplace constitution}
(rules a platform or exchange can adopt unilaterally and commit to), the
\emph{regulator} (mandates that no equilibrium actor supplies privately),
and \emph{infrastructure} (credentials and benchmarks that make the other
two contractible). The failure--instrument matrix
(Table~\ref{tab:matrix}) assigns cures to diseases;
Section~\ref{sec:implications} states the general lessons;
Section~\ref{sec:actors} assigns the work.

\begin{table*}[t]
\caption{The failure--instrument matrix. Each row is a derived assignment:
the failure in column two is not repaired by the instruments in column
three acting alone, and is matched to the instrument in column four by the
result in column five.}
\label{tab:matrix}
\centering
\begin{tabular}{c p{3.6cm} p{3.4cm} p{5.6cm} p{2.6cm}}
\hline
\# & Failure & Not fixable by (alone) & Matched instrument & Source \\
\hline
1 & Fabricated lineage & performance scores; consent badges & provenance
credential (existing TEE infrastructure) & A.\ref{as:provenance}
(presupposed) \\
2 & Low fidelity; poor buyer--task--host fit & provenance; any global
score & performance certification; compatibility-class tests; sandbox and
audited replay & Lem.~\ref{lem:no-score}; \S\ref{sec:extensions} \\
3 & Unauthorized extraction, sale, or use & provenance; performance &
scoped license $\ell$: purpose, recipient, horizon, call budget,
revocation & Def.~\ref{def:license} \\
4 & Unpriced subject exposure & price and reputation & coverage
$s^{*}=\lambda+w_{S}^{\star}/n$, indexed per quality-deployment &
Thm.~\ref{thm:secondbest}(i); Prop.~\ref{prop:selection}(iii) \\
5 & Refresh underinvestment on either margin & reassigning ``the owner'' &
interior renewable claims $(w_{S}^{\star},w_{P}^{\star})$; control and
cash flow set separately & Thm.~\ref{thm:rights};
Cor.~\ref{cor:stake}; Prop.~\ref{prop:corners} \\
6 & Long-run depletion from reach expansion & nonrivalry intuition; access
metrics & scope haircut $n^{*}\approx\tfrac23(1-\lambda)$; sequence
portability with coverage & Thm.~\ref{thm:headline};
Thm.~\ref{thm:secondbest}(iv) \\
7 & Behavioral pollution; obfuscation equilibrium & transparency alone;
compensation alone & transparency \emph{and} coverage floor jointly;
rule-based commitment & Thm.~\ref{thm:pollution} \\
8 & Coverage dilution as the market grows & share-financed compensation &
royalty-financed coverage ($\tau$), reach-invariant &
Prop.~\ref{prop:collapse}(ii) \\
9 & Over-participation and relational harm & focal-subject consent &
relational-density monitoring \eqref{eq:wedges}; affected-party and
collective governance & Thm.~\ref{thm:formation}; Cor.~\ref{cor:bridge} \\
10 & Self-perpetuating de facto control & renaming the owner & portability
$\phi$ and salience $\zeta$ mandates; delegation per the feasibility map &
Cor.~\ref{cor:selfperp}; Thm.~\ref{thm:delegation} \\
\hline
\end{tabular}
\end{table*}

Two structural properties of the matrix deserve emphasis. \emph{No column
dominates}: every instrument in the fourth column fails against at least
one failure other than its own, which is why the matrix cannot be
collapsed into a single score or a single reform
(Section~\ref{sec:nofusion}). \emph{The failure column is ordered by
visibility}: failures 1--4 are legible in a single transaction; failures
5--10 are equilibrium objects, visible only in aggregates (refresh rates,
stock trajectories, pool composition, participation densities). That is
why Section~\ref{sec:empirics}'s measurement interface is part of
governance, not an academic afterthought.

\subsection{Ten derived implications}\label{sec:implications}

\begin{enumerate}
\item \textbf{There is no universal owner}
(Thm.~\ref{thm:rights}; Prop.~\ref{prop:corners}). Two noncontractible
investors both need marginal incentives; budget balance forbids giving both
the full margin; and each ownership corner destroys the other party's
supply. ``Who owns the memory'' is the wrong regulatory question.
\item \textbf{Renewable claims are conditional, and the condition is
estimable} (Cor.~\ref{cor:stake}; Thm.~\ref{thm:secondbest}(ii)). The
subject should hold a positive continuation claim iff
$k_{S}/k_{P}\ge C/\Omega$; this requires no operational control and is an
econometric threshold, not an ideological one.
\item \textbf{Compensation is quality policy}
(Thm.~\ref{thm:pollution}(ii); Thm.~\ref{thm:secondbest}(ii)). Paying
subjects buys authenticity (an efficiency argument independent of
fairness), and in the second best the authenticity floor is
self-enforcing wherever the subject is a productive refresher.
\item \textbf{Transparency and compensation are complements}
(Thm.~\ref{thm:pollution}(iii)). Mandating either alone converts naive
supply into distorted supply or leaves the naive exposed; only jointly do
authenticity and voluntariness recover.
\item \textbf{How compensation is financed matters as much as its level}
(Prop.~\ref{prop:collapse}(ii)). Share-financed coverage dilutes
mechanically as reach expands and crosses the authenticity floor at twice
the subject-preferred scope; growing markets should be royalty-financed.
\item \textbf{Access metrics are not welfare metrics}
(Thm.~\ref{thm:headline}; Cor.~\ref{cor:access}). Reach-expanding policy
can raise measured access while shrinking the long-run stock; evaluation of
portability and interoperability mandates must include the refresh margin.
\item \textbf{Participation volume is not a health metric}
(Thm.~\ref{thm:formation}). The marginal lister carries one negative and
two positive externalities; the regulator's object is the three-wedge
comparison \eqref{eq:wedges} (relational density against appropriability
and thickness), stratified by memory class, not listing counts.
\item \textbf{Ownership reform without feasibility infrastructure is
inert} (Thm.~\ref{thm:rights}(ii); Cor.~\ref{cor:selfperp};
Thm.~\ref{thm:delegation}). Renaming the owner while portability $\phi$
and salience $\zeta$ stay low changes neither scope nor refresh, and the
incumbent maintains those conditions in equilibrium; $\phi$ and $\zeta$
are the binding policy margins.
\item \textbf{The scope rule is robust; the compensation rule is
information-hungry} (Thm.~\ref{thm:secondbest}(iii)--(iv)). $n^{*}$ sits
in a four-point band around two-thirds of the static optimum regardless of
the productivity split, so scope can be set with an estimate of $\lambda$
alone; coverage requires the full decomposition $(k_{S},k_{P})$.
\item \textbf{Monetary IR does not exhaust consent}
(A.\ref{as:commitment}; Rem.~\ref{rem:normative}). For inalienable uses
the constraint is $\ell\in\mathcal R_{i}(m)$, not a higher royalty; and
where pure efficiency logic recommends tolerated pollution
(Prop.~\ref{prop:regimes}(c)), that recommendation is exactly what
normative scrutiny should target rather than inherit.
\end{enumerate}

\subsection{Who implements what}\label{sec:actors}

The five components of Section~\ref{sec:constitution}'s constitution divide
cleanly across the three actors. The \emph{marketplace} can adopt
unilaterally: quality-indexed coverage at $s^{*}$ with rule-based
commitment (Theorem~\ref{thm:pollution}(iv) shows commitment is in its own
interest when $(1-\nu)RQ>c_{d}$); scoped licenses and audited replay;
compatibility-class evaluation; royalty financing. The \emph{regulator}
must supply what no equilibrium actor will: the transparency mandate
$\mathsf T$ (Theorem~\ref{thm:pollution}(iii)); portability and salience
infrastructure (Corollary~\ref{cor:selfperp}); relational-density
monitoring and, for high-$\varrho$ classes, participation review
(Theorem~\ref{thm:formation}); scope defaults or deployment subsidies where
delegation is infeasible (Theorem~\ref{thm:delegation}(i)); and the
inalienable-use boundary itself. \emph{Infrastructure} (TEE provenance,
class-indexed benchmarks, audit logs) is what makes the first two
contractible; its role is enabling, not substitutive: certification serves
authorization (Remark~\ref{rem:why-tau}). The assignment explains a pattern
the matrix alone does not: the instruments that fix \emph{visible} failures
(rows 1--4) are mostly marketplace-implementable, while the instruments
that fix \emph{equilibrium} failures (rows 5--10) almost all require the
regulator or a standards body. Markets for person-linked memory
self-correct their transactions, not their dynamics.

\subsection{The no-fusion rule}\label{sec:nofusion}

The matrix compresses into one design rule: \textbf{a marketplace should
not issue a fused ``memory trust score.''} The model now supplies three
independent reasons for it. First, \emph{crossing rankings}: no
buyer-independent scalar orders artifacts correctly for all buyers
(Lemma~\ref{lem:no-score}), so any global score misprices fit. Second,
\emph{direction-flipping}: the information content of an authorization
badge depends on the harm technology's sign and the royalty structure
(Proposition~\ref{prop:selection}); a fused score would hard-code one
sign and mislead whenever the other obtains. Third,
\emph{regime-dependence}: whether high measured coverage signals health or
a binding floor, and whether thick participation signals success or
relational over-supply, depends on parameters
(Proposition~\ref{prop:regimes}; Theorem~\ref{thm:formation}) that a
scalar cannot carry. Provenance credentials, performance and fit reports,
and permission credentials answer different questions, fail differently,
and, by the results above, actively mislead when merged. The minimal
credential dimensionality claim is formalized as a conjecture in
Section~\ref{sec:limits}.

\section{Measurement, Calibration, and Evidence}\label{sec:empirics}

\subsection{What the model exposes}\label{sec:estimands}

The model was written so that its primitives are measurable rather than
postulated, and Section~\ref{sec:governance} showed that most of its
governance content is a comparison between measured quantities: the
productivity ratio $k_{S}/k_{P}$ against $C/\Omega$
(Theorem~\ref{thm:secondbest}(ii)), the three wedges of \eqref{eq:wedges},
the sign of $dY^{E}/dn$ in \eqref{eq:depletion}. A paper that ends at the
theorems therefore leaves its own design rules unusable. This section closes
that gap in three layers: a numerical companion that maps the parameter
regions (Section~\ref{sec:companion}); a synthetic memory economy in which
every estimand has a known ground truth, so that estimators can be
\emph{validated} and the theorems can actually fail
(Sections~\ref{sec:testbed}--\ref{sec:eq-results}); and pre-registered
designs for platforms holding real interaction data
(Section~\ref{sec:prereg}). The middle layer is the one no observational
study can supply: on real data one never sees $q_{u}$'s ground truth, so an
estimate of $k_{S}$ cannot be checked against anything.

\subsection{Layer 1: the parameter map}\label{sec:companion}

\begin{figure}[t]
\centering
\includegraphics[width=\columnwidth]{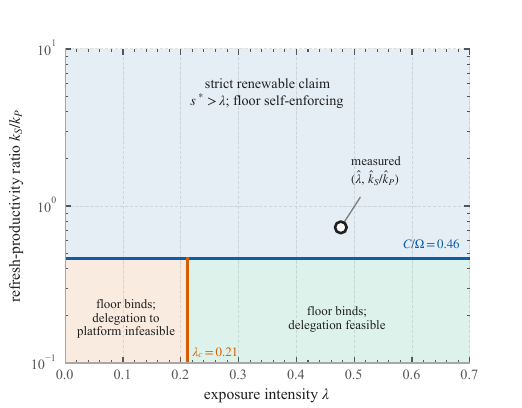}
\caption{The institutional parameter map. The horizontal boundary is
Theorem~\ref{thm:secondbest}(ii)'s self-enforcing coverage floor evaluated at
the second-best scope, $k_{S}/k_{P}=C/\Omega=(3-\sqrt3)/(1+\sqrt3)\approx
0.46$, independent of $\lambda$, because $n^{*}$ scales with $1-\lambda$.
The vertical boundary is Theorem~\ref{thm:delegation}(i)'s delegation
threshold $\lambda_{c}=(3-\sqrt3)/6\approx0.21$. The marker is the point
recovered by the testbed of Section~\ref{sec:testbed}.}
\label{fig:regimes}
\end{figure}

Figure~\ref{fig:regimes} plots the two boundaries that decide which
instruments a memory class needs. Above the horizontal line the subject is a
productive enough refresher that the second best pays her above bare exposure
without any authenticity mandate; below it, the floor binds and
Proposition~\ref{prop:regimes}'s regimes apply. Left of the vertical line the
de facto platform-controlled regime cannot be steered to the second best with
nonnegative royalties, so scope regulation or deployment subsidies are
required. Two measured numbers place a memory class on this map, and
nothing else about the class matters for the choice of instrument, which is
the practical content of Sections~\ref{sec:optimal}--\ref{sec:governance}.

\subsection{Layer 2: a synthetic memory economy}\label{sec:testbed}

We build an environment that instantiates Definitions~\ref{def:state},
\ref{def:artifact}, \ref{def:fidelity}, \ref{def:exposure} and
\ref{def:kappa} exactly. Subject $i$'s state is
$\omega_{i}=(\omega^{c},\omega^{g}_{k(i)},\omega^{p}_{i})$ over disjoint
coordinate blocks, each block an AR(1) with its own persistence
($\rho_{c}=0.995$, $\rho_{g}=0.97$, $\rho=0.85$): every component drifts,
and at different rates. Interaction generates
$y=\langle\omega,x\rangle+\nu$ with $x\sim N(0,I)$; \emph{subject effort} $e$
raises the informativeness of each interaction,
$\sigma^{2}_{\nu}(e)=\sigma^{2}_{\nu 0}/(1+a_{S}e)$, and
\emph{platform effort} $x_{P}$ raises the share of the history's information
that extraction retains, $\psi(x_{P})=1-e^{-a_{P}x_{P}}$. The memory artifact
is the resulting Gaussian belief; its stationary information stock has the
closed form $J^{*}_{jj}=\psi F/(1-\rho_{j}^{2})$.

Three properties make this environment the right instrument. First, fidelity
is \emph{exact}: with quadratic loss, $V_{u}(\omega)=0$, so
Definition~\ref{def:fidelity} reduces to a variance-reduction ratio,
$q_{u}=1-w^{\top}\Sigma_{\varphi}w/w^{\top}\Sigma_{0}w$, computed rather than
estimated. Second, the extraction ceiling is a \emph{test} rather than an
assumption: $\psi$ is a free parameter that can be pushed to the
sufficient-statistic limit. Third, demand is \emph{derived}: buyers differ in
a use's personal-loading share $\chi_{b}$ and in intensity $A_{b}$, and
$R(n),C(n),B(n)$ are read off the realized value distribution. The resulting
demand curve is not the uniform benchmark's (regressing derived $R$ on
$n(1-n)$ gives a negative $R^{2}$, and derived $R$ peaks at $n\approx0.64$
rather than $\tfrac12$), so agreement between the structural optimum and
the closed forms of Sections~\ref{sec:results}--\ref{sec:optimal} is
informative rather than tautological. Likewise $\lambda$ is derived, as
$\lambda_{0}$ times the adversarial-to-buyer decision-value ratio the
environment produces.

Every result below is reported as a family rather than a point. The
parameters that were set by hand (state persistence $\rho$, the dimension
$d_{p}$ of the personal block, the fidelity regime, buyer-intensity
dispersion, the cost ratio $c_{S}/c_{P}$, and the seed) are perturbed, and
each figure plots the resulting family. This separates structural conclusions
from calibration artifacts, and in two cases the perturbation itself
identifies the mechanism behind a result.

\subsection{What can fail, and what cannot}\label{sec:what-can-fail}

A synthetic testbed can produce three kinds of statement, and conflating them
is the standard way computational sections overclaim. We label each result in
Table~\ref{tab:results} accordingly.

An \emph{identity check} cannot fail: it is true by construction in the
Gaussian environment, and what it verifies is the implementation, not the
theory. Proposition~\ref{prop:ceiling} restricted to garblings of the
subject's own history is of this kind (extraction there is a scalar
contraction of the information matrix, so the ordering follows from positive
semi-definite monotonicity), as is the $\rho^{2k}$ decay of
Corollary~\ref{cor:staleness}. We report these because an implementation that
failed them would be wrong, not because passing them is evidence.

A \emph{test} has a direction in which it fails, and some of ours do. The
informative form of the ceiling is therefore not the in-class check but its
complement: an extractor that is \emph{not} a garbling of the subject's
history (one that imputes from a population prior, say) breaches the
ceiling in $68\%$ of checks, by up to $0.045$ in fidelity units. That is the
content of Proposition~\ref{prop:ceiling}: it is a statement about what
memory extraction cannot do, and it fails precisely for artifacts that are
not, in the definitional sense, memory. The entanglement conjecture, the
audit bound, the depletion condition and the scope haircut are all tests in
this sense, and one of them fails outright below.

An \emph{illustration} evaluates a closed form or an imposed rule rather than
an equilibrium. Theorem~\ref{thm:pollution}'s distortion $d^{*}$ and
Proposition~\ref{prop:selection}'s authorization rule are computed, not
solved for; ``$d^{*}=0$ iff $s\ge\lambda$'' is then an arithmetic property of
the expression, not a finding. We keep them as consistency checks and label
them as such. What the environment does supply there is the pass-through:
each unit of distortion destroys $0.16$ units of realized quality, which the
closed form does not say.

\subsection{Identification results}\label{sec:id-results}

\begin{table}[t]
\caption{Layer-2 results, reported over the perturbation sweeps rather than
at a single calibration. ``Violations'' counts strict violations of the
stated inequality across the whole sweep. Rows marked \textsc{illustration}
evaluate a closed form or an imposed rule rather than an equilibrium and
cannot fail; see Section~\ref{sec:what-can-fail}. All figures are read from
the committed \texttt{results/*.json}.}
\label{tab:results}
\centering
\begin{tabular}{p{2.35cm} p{2.55cm} p{2.55cm}}
\hline
Target & Test & Result \\
\hline
Prop.~\ref{prop:ceiling}\newline extraction ceiling
& $q_{u}(\varphi)\le q_{u}(h^{t})$ over 11{,}200 (use, budget) pairs,
  in class and out
& 0 violations in class; gap $\to3.1\times10^{-7}$ as $\psi\to1$;
  breached in $68\%$ of out-of-class checks \\[2pt]
Cor.~\ref{cor:staleness}\newline staleness
& personal-component decay vs.\ $\rho^{2k}$
& max error $1.6\times10^{-15}$ over $\rho\in[0.75,0.95]$ \\[2pt]
Rem.~\ref{rem:delta-bridge}\newline $\delta$ is derived
& $\hat\delta$ across use classes $\chi\in[0,1]$
& $8.8\times$ range at $\rho=0.85$; $13.8\times$ at $\rho=0.75$ \\[2pt]
Lem.~\ref{lem:refresh}\newline refresh technology
& fit $Y=\eta_{S}e+\eta_{P}x+\gamma ex$
& $R^{2}=0.82$; $\hat k_{S}/\hat k_{P}=0.73$; $\gamma<0$ at
  $\bar q=0.81$ but $>0$ under saturation \\[2pt]
Conj.\ (Gaussian frontier)
& $\kappa^{\star}$ vs.\ $\chi$, aligned adversary, four fidelity regimes
& \emph{fails} at low fidelity (error $0.38$); error
  $\approx0.36\,(1-\bar q)$; $<2\times10^{-3}$ once $\bar q>0.99$ \\[2pt]
Lem.~\ref{lem:kappa-iota}\newline $\kappa^{\star}\le\iota$
& optimal garbling vs.\ feasible identity replacement
& 0 violations; audit overstates by $65\%$--$91\%$, rising in $d_{p}$ \\[2pt]
Thm.~\ref{thm:headline}\newline depletion \eqref{eq:depletion}
& sign of $d\bar Q/dn$ at fixed $(\alpha,\tau)$, derived demand
& $86.4\%$ mean agreement over 36 calibrations, $83.9\%$ ungated;
  $51\%$ worst cell \\[2pt]
Thm.~\ref{thm:secondbest}(ii)\newline coverage floor
& $\operatorname{sign}(s^{*}-\lambda)$ vs.\ $k_{S}/k_{P}\gtrless C/\Omega$
& 10 of 12 cells at the structural $k$-ratio; 8 of 12 at the
  \emph{estimated} one \\[2pt]
Thm.~\ref{thm:pollution}(i)\newline pollution
& \textsc{illustration}: $d^{*}$ evaluated from the closed form, not
  from an equilibrium
& consistency check only; measured quality pass-through $0.16$ per unit
  of distortion \\[2pt]
Prop.~\ref{prop:selection}\newline selection
& \textsc{illustration}: $\mathrm{corr}(\tilde q,\text{authorized})$
  under an imposed authorization rule, flat vs.\ indexed
& flat: $-0.45$ to $-0.24$; indexed: $|r|\le0.03$ throughout \\
\hline
\end{tabular}
\end{table}

\begin{figure*}[t]
\centering
\includegraphics[width=\textwidth]{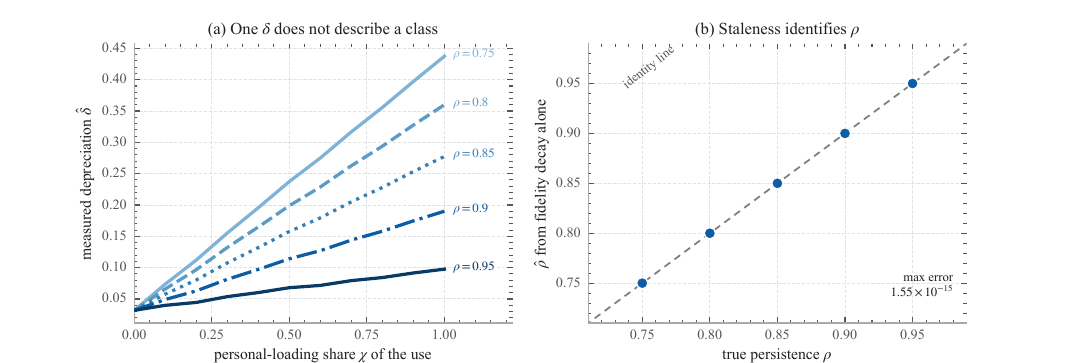}
\caption{Staleness. (a) Measured one-period depreciation against the use's
personal-loading share $\chi$, for five values of the state persistence
$\rho$: a single $\delta$ describes a memory class only if every buyer loads
identically on the personal component, and the spread widens as drift rises.
(b) Persistence recovered from fidelity decay alone, against its true value;
the maximum error over the five calibrations is $1.4\times10^{-15}$.}
\label{fig:staleness}
\end{figure*}

\begin{figure*}[t]
\centering
\includegraphics[width=\textwidth]{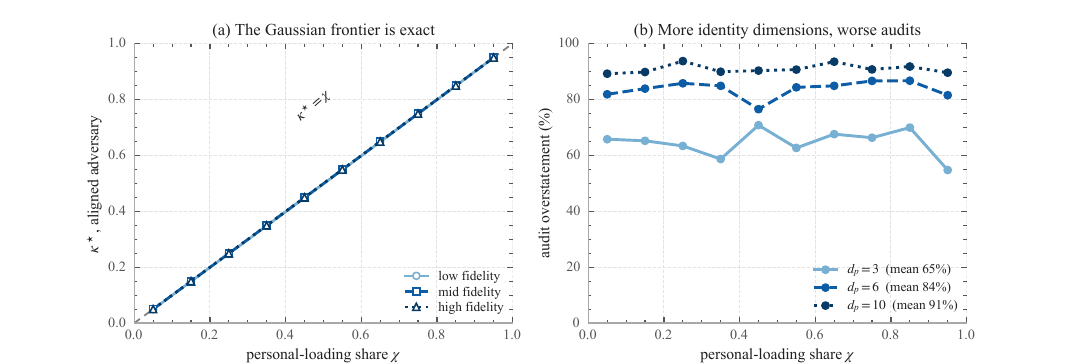}
\caption{Entanglement. (a) Against an adversary aligned with the buyer's
personal loading, $\kappa^{\star}=\chi$ on the diagonal in every fidelity
regime: the three curves coincide, which is the result. (b) The gap
between an identity-replacement audit and the true entanglement, as a
percentage of the audit, against $\chi$ and across the dimension $d_{p}$ of
the personal block. The overstatement rises with $d_{p}$ because a targeted
garbling has more directions to preserve while naive erasure destroys them
all.}
\label{fig:entanglement}
\end{figure*}

\begin{figure*}[t]
\centering
\includegraphics[width=\textwidth]{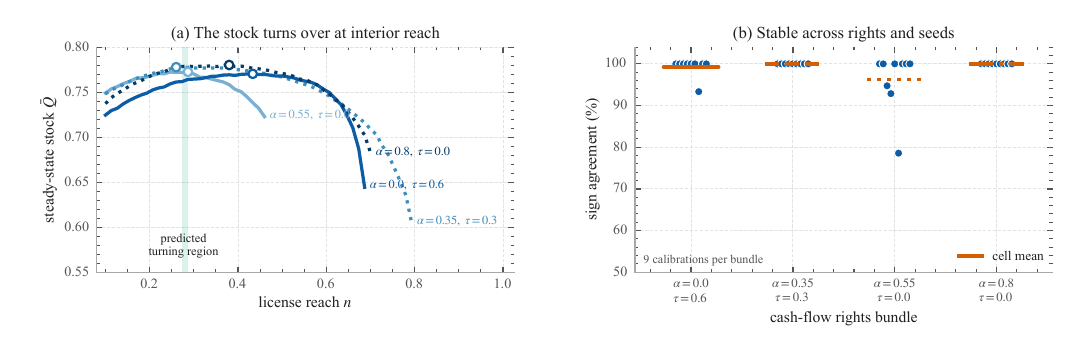}
\caption{Depletion. (a) Steady-state stock against license reach at four
fixed cash-flow bundles, one representative calibration: the stock rises,
turns over at an interior reach, and falls; that is
Theorem~\ref{thm:headline}'s mechanism, with the region where the predicted
derivative changes sign shaded. (b) Sign agreement between the predicted and
realized derivative, over the interior region where the prediction is
non-vacuous, for each bundle across nine calibrations (three dispersions
$\times$ three seeds).}
\label{fig:depletion}
\end{figure*}

\begin{figure*}[t]
\centering
\includegraphics[width=\textwidth]{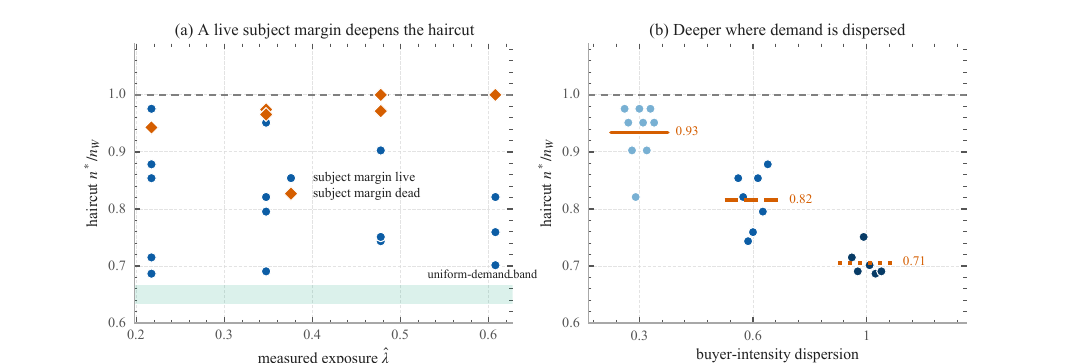}
\caption{The scope haircut. (a) $n^{*}/n_{W}$ against measured exposure,
separating calibrations in which the subject's refresh margin is live at the
optimum from those in which it is not; the uniform-demand band of
Theorem~\ref{thm:secondbest}(iv) is shaded. (b) The same haircuts grouped by
buyer-intensity dispersion, live margins only: the haircut deepens as demand
becomes more disperse, and the uniform band is approached only at the
dispersed end.}
\label{fig:haircut}
\end{figure*}

\begin{figure*}[t]
\centering
\includegraphics[width=\textwidth]{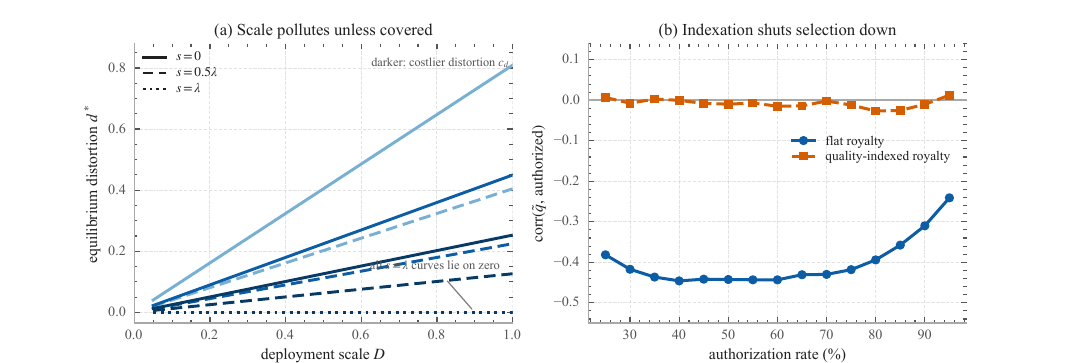}
\caption{Pollution and selection. (a) Equilibrium distortion against
deployment scale, for three coverage levels and three distortion costs:
distortion grows with scale whenever coverage falls short of exposure and is
identically zero once it does not. (b) The correlation between fidelity and
authorization, against the authorization rate, under a flat and under a
quality-indexed royalty.}
\label{fig:pollution}
\end{figure*}

Table~\ref{tab:results} collects the results;
Figures~\ref{fig:staleness}--\ref{fig:pollution} plot each with its
perturbation family.

Proposition~\ref{prop:ceiling} survives 11{,}200 checks with no violation and
binds with equality in the sufficient-statistic limit, as stated. Corollary~
\ref{cor:staleness} holds to machine precision, and the persistence parameter
is recovered exactly from fidelity decay alone ($\hat\rho=0.85000$ against
a truth of $0.85$), which is the identification argument behind the
class-specific staleness estimates Section~\ref{sec:governance} asks
regulators to collect, and Figure~\ref{fig:staleness}(b) shows the recovery
is exact across the whole range of persistence we tried, not at one
calibration. Figure~\ref{fig:staleness}(a) makes
Remark~\ref{rem:delta-bridge} quantitative: at $\rho=0.85$ measured
one-period depreciation runs from $0.032$ for a use loading entirely on the
common component to $0.278$ for a purely personal one, and the spread widens
as drift rises: a $13.8\times$ range at $\rho=0.75$. A single $\delta$ for
``memory'' is therefore not a summary of anything; depreciation must be
estimated per use class, and the same artifact is fresh for one buyer and
stale for another.

Two findings amend the paper. First, the Gaussian entanglement conjecture is
\emph{not} an identity: it is a high-fidelity limit, and it fails cleanly
when the memory is far from the information ceiling
(Figure~\ref{fig:entanglement}(a)). Against an adversary aligned to the
buyer's personal loading, $\kappa^{\star}$ approaches $\chi$ from below as
the artifact becomes informative, with a maximum error of $0.38$ at
$\bar q=0.34$, $0.16$ at $\bar q=0.67$, $9.0\times10^{-3}$ at $\bar q=0.98$
and $6.0\times10^{-4}$ at $\bar q=1.00$. The convergence is first order in
the fidelity gap: a log--log fit of $|\kappa^{\star}-\chi|$ on $1-\bar q$
returns a slope of $1.02$, so the defensible statement is
\begin{equation}\label{eq:conj1-limit}
\kappa^{\star}_{\text{aligned}}=\chi+O(1-\bar q),
\qquad |\kappa^{\star}-\chi|\approx0.36\,(1-\bar q).
\end{equation}
This replaces a claim we made in an earlier draft, and the way it was
produced is worth stating plainly, because it is the kind of error a
synthetic testbed invites. The first version of this experiment separated
``fidelity regimes'' by the extraction budget alone while holding the
information environment at its default, where the stationary stock sits above
$q=0.99$ regardless: all three nominal regimes returned
$\bar q\in[0.990,0.9997]$, the error was $2.2\times10^{-3}$ everywhere, and
the conjecture looked exact in every regime; it was never tested outside
saturation. Separating the regimes by the information environment instead
(the number of interactions and the observation noise, which is what
actually moves the stock) reproduces the failure at low fidelity in two
independently implemented pipelines. The conjecture should be stated as a
proposition qualified by \emph{both} the adversary's alignment and an
information-saturation limit, with \eqref{eq:conj1-limit} as its rate. The
practical reading is more useful than the identity would have been: the
Gaussian frontier is a good approximation exactly for the mature, heavily
refreshed artifacts that are worth trading, and a poor one for thin memory,
where an aligned adversary extracts materially less than the personal
loading share suggests.

Second, and more consequentially, $\kappa^{\star}$ is sharply
\emph{adversary-indexed}. Lemma~\ref{lem:kappa-iota} holds (no violations
anywhere in the sweep), but its bound is slack by an order of magnitude for
any adversary that targets a specific attribute rather than the whole
identity: identity-replacement audits overstate entanglement by $65\%$ to
$91\%$, because erasing the personal coordinates destroys the entire personal
block while the optimal zero-exposure garbling need only remove one direction
of it. The comparison is only meaningful against a \emph{feasible} erasure.
Withholding the personal coordinates leaves the Schur complement of the
personal block, not the raw non-personal sub-matrix; zeroing rows and columns
instead retains more information than any garbling can deliver, and produces
spurious violations of Lemma~\ref{lem:kappa-iota}: $0.117$ in these
regimes, against $0$ once the erasure is made feasible. (For diagonal
information matrices the two coincide, which is why nothing computed from the
closed-form stationary stock is affected.) The perturbation identifies the
mechanism rather than merely sizing the gap
(Figure~\ref{fig:entanglement}(b)): the overstatement rises monotonically in
the dimension of the personal block ($65\%$ at $d_{p}=3$, $84\%$ at
$d_{p}=6$, $91\%$ at $d_{p}=10$), which is what the garbling program
predicts, since each additional identity dimension is one more direction a
targeted release can keep. Richer personal state makes naive erasure worse,
not safer. Remark~\ref{rem:iota-upper}'s caution that audits are an upper
bound is correct but far too mild. The operational implication is that
$\kappa^{\star}$ must be reported against a \emph{named adversary class};
an unqualified $\kappa^{\star}$ near $1$, read off an identity-replacement
audit, will justify institutional pessimism (``no privacy technology can
improve the frontier'') in classes where a targeted garbling would in fact
retain most of the buyer value.

A third finding concerns the refresh technology itself, and it is a finding
about the \emph{estimator} rather than about the technology. At the
calibration where the refresh margin is interior ($\bar q=0.81$) the fitted
interaction term $\gamma$ is \emph{negative}, about $12\%$ of the linear
terms' magnitude, because the true technology is concave in the
multiplicative information index $\psi(x)F(e)$, which a
linear-plus-interaction specification reads as mild substitutability.
Assumption~\ref{as:forms}'s restriction $\gamma\ge0$ is therefore not
innocuous as a description, though it is harmless as a benchmark: every
headline result is stated for the additive case $\gamma=0$, which the
measurement supports as the right neutral specification. The assumption
should be relaxed to allow either sign, and the reason is sharper than a
single sign. Re-run at the information-rich calibration, where steady-state
fidelity exceeds $0.99$, the same regression attenuates $\hat\eta_{S}$ by a
factor of $48$ \emph{and $\gamma$ flips positive}. The sign of the
interaction term is not a property of the refresh technology at all; it is a
property of where on the saturation curve the economy is sitting when the
analyst measures it.

The same attenuation contaminates the productivity ratio the paper's
institutional prescriptions turn on. The environment implements a structural
ratio $(a_{S}^{2}/c_{S})/(a_{P}^{2}/c_{P})$; Lemma~\ref{lem:refresh}'s
regression recovers $\hat\eta^{2}/c$, and these are not the same number.
Across the cost sweep the estimated ratio is uniformly $3.80\times$ below the
structural one ($0.24$ against $0.93$, $0.73$ against $2.78$, $2.44$
against $9.26$), with a strikingly constant attenuation factor, which is
itself a usable diagnostic. Because the ratio enters
Theorem~\ref{thm:secondbest}(ii) through a comparison with $C/\Omega$, the
bias is not innocuous: the theorem's prescription is recovered in 10 of 12
cells at the structural ratio but only 8 of 12 at the estimated one, and
every one of the extra failures is a case where the estimate falls below a
threshold the truth clears. This qualifies the section's own claim: what the
testbed validates is that the paper's estimators are consistent \emph{at an
interior operating point}; it also shows that they are biased toward
understating the subject's refresh productivity as fidelity rises, which
would bias real institutional advice toward too little coverage. Any
empirical estimate of $(k_{S},k_{P})$ must be reported together with the
fidelity level at which it was taken, or a saturated system will look like
one with no refresh margin at all.

\subsection{Equilibrium results}\label{sec:eq-results}

The depletion condition \eqref{eq:depletion} is the paper's headline and the
test most able to fail, since demand here is not the uniform benchmark's.
Theorem~\ref{thm:headline} requires that scope be swept at fixed cash-flow
rights $(\alpha,\tau)$, and Proposition~\ref{prop:collapse}(ii) makes this
essential, since fixed rights put coverage on a declining path. Sweeping
scope that way reproduces the mechanism directly
(Figure~\ref{fig:depletion}(a)): the steady-state stock rises with reach,
turns over at an interior scope, and falls, with the turning point moving
across rights bundles. Over 36 calibrations (four bundles $\times$ three
dispersions $\times$ three seeds) the predicted and realized signs of
$d\bar Q/dn$ agree on $86.4\%$ of the interior points where the predicted
derivative is non-negligible, and on $83.9\%$ unconditionally, with a worst
cell of $51\%$ (Figure~\ref{fig:depletion}(b)). The gate matters and we state
it: ``non-negligible'' means a predicted derivative at least $5\%$ of the
cell's largest, which retains $94\%$ of the live points. An earlier version
gated on the upper half of $|\text{pred}|$ within each cell, which is a rank
rule (it discards half of every cell whatever the magnitudes are), and
returned $98.9\%$. That figure was an artifact of the gate, not a property of
the mechanism, and the honest number is the low-to-mid eighties either way.
The condition survives a demand curve it was not derived for, and survives it
at every seed, but it is a mechanism that predicts the sign of a derivative
correctly five times in six, not nineteen times in twenty.

Theorem~\ref{thm:secondbest}(ii)'s coverage floor is correct in 10 of 12
cells of a $\lambda\times(k_{S}/k_{P})$ sweep when the productivity ratio is
the structural one, and in 8 of 12 when it is the ratio an analyst would
estimate; the difference is the attenuation discussed above. Both failures in
the structural case occur where $C/\Omega$ is near or above $1$, i.e.\ where
exposure consumes most of the social stake and the comparison is delicate,
which is a boundary the theorem should state.

Theorem~\ref{thm:secondbest}(iv)'s scope haircut requires the most
substantial amendment. Under derived demand the haircut is real but its
magnitude is not the uniform band: $n^{*}/n_{W}=0.83$ on average over 36
calibrations, ranging $0.69$--$0.98$. Its \emph{existence}, however, is
sharply conditional in a way the closed form does not say.
Figure~\ref{fig:haircut}(a) shows the haircut appearing exactly when the
subject's refresh margin is live, and vanishing (mean $0.99$, and exactly
$1.000$ in nine of fourteen such cells) whenever the optimum leaves
$w_{S}\le0$. Because the amendment turns on a binary classification, it is
worth saying that the classification is not delicate: the split is structural
($w_{S}>0$ at the optimum) rather than a threshold on a discretized effort
grid, and every cell classified as live sits at least $29$ grid steps (an
effort of $0.22$) away from the corner. Nothing here is knife-edge in the
argmax. This is the correct statement of the result and it is stronger
than the band: \emph{the second best withholds reach only to protect a
refresh margin that exists}, so the haircut is not a universal constant but
the shadow price of a live subject stake. The perturbation also says what the
magnitude depends on (Figure~\ref{fig:haircut}(b)): the haircut deepens
monotonically in buyer-intensity dispersion ($0.93$, $0.82$, $0.71$ across
the three dispersion levels), and the uniform-demand band is approached
only at the dispersed end, which is the sense in which
$[0.63,0.67]$ is a special case rather than a constant. Where the
platform is the overwhelmingly dominant refresher, the static and dynamic
scope rules coincide, and the case for restricting reach on depletion grounds
disappears with the margin it was protecting. Theorem~\ref{thm:secondbest}(iv)
should be restated directionally, with $[0.63,0.67]$ presented as the
uniform-demand special case it is.

Finally, Theorem~\ref{thm:pollution}(i) and
Proposition~\ref{prop:selection} are illustrated rather than tested, and we
mark them so (Figure~\ref{fig:pollution}). The distortion $d^{*}$ is
evaluated from the theorem's own closed form rather than solved for as an
equilibrium of the subject's distorted problem, so ``distortion is zero
exactly when coverage reaches the floor, and rises linearly in deployment
scale otherwise'' is a property of that expression and holds at every
distortion cost by construction. What the environment adds is the
consequence: feeding $d^{*}$ back through the synthetic economy, each unit of
distortion costs $0.16$ units of realized quality, so the pollution channel
is economically material and not merely formal. Solving the distorted
equilibrium properly is the obvious next iteration of the testbed. The
selection direction
flips as predicted with the royalty structure and stays flipped across the
whole range of authorization rates: under a flat royalty the correlation
between fidelity and authorization runs from $-0.45$ to $-0.24$, while under
quality-indexed royalties it never leaves $\pm0.03$.
Proposition~\ref{prop:selection}(iii) is thus not a knife-edge: indexation
shuts selection down cleanly, whatever the royalty level.

\subsection{Layer 3: designs for platforms}\label{sec:prereg}

Layer 2 validates estimators; it cannot establish external magnitudes. Three
designs transfer directly to a platform holding real interaction data, and
each is stated so that a null result is interpretable. (i) \emph{Refresh
decomposition}: randomize prompts for subject feedback (corrections,
confirmations, preference edits) against engineering interventions
(consolidation cadence, retrieval budget) in a factorial design; regress
held-out task performance on both arms to recover $(\hat\eta_{S},
\hat\eta_{P},\hat\gamma)$ and hence $k_{S}/k_{P}$, reporting the fidelity
level at which the estimate is taken. (ii) \emph{Scope randomization}:
randomize license reach $n$ across matched memory cohorts at fixed
$(\alpha,\tau)$ and track subsequent refresh flow, testing
\eqref{eq:depletion} directly; the design must hold cash-flow rights fixed
rather than coverage, or the mechanism is assumed away. (iii)
\emph{Royalty-structure randomization}: assign flat versus quality-indexed
royalties across subject cohorts and compare the fidelity composition of the
authorized pools, testing Proposition~\ref{prop:selection}. For
\eqref{eq:wedges} the three wedges are separately measurable: relational
density from within-class predictive leakage onto non-listing peers,
appropriability from the coverage share, and thickness from the price
response to market size. Theorem~\ref{thm:formation}'s direction is thus an
empirical comparison, not a presumption. Each design should pre-register the
estimand and the sign that would falsify the corresponding result.

\subsection{Threats, scope, and replication}\label{sec:threats}

The Gaussian-linear environment buys exactness at the cost of realism: real
memory is not a belief over a linear-Gaussian state, real extraction is not a
precision-scaling operator, and real adversaries are not linear functionals.
What Layer 2 establishes is therefore internal: that the paper's estimators
recover what they claim to recover when ground truth exists, and that its
theorems are not artifacts of the uniform-demand parametrization. The
perturbation sweeps let us say which is which. Four conclusions are stable
across every parameter we varied and are properties of the class of
environments rather than of this one: the extraction ceiling (a
data-processing inequality), the $\rho^{2k}$ staleness law (any AR(1)
state), the sign of the depletion condition, and the selection reversal.
Three are directional but not numerical: the haircut always exists when the
subject margin is live and always vanishes when it is not, but its depth
moves with demand dispersion; the audit overstatement is always large but
runs from $65\%$ to $91\%$ with the dimension of identity; and $\delta$'s
spread across use classes is always wide but scales with drift. Point
magnitudes should be read as demonstrations that these quantities are large
enough to matter, not as estimates of their real-world values.

One conclusion is explicitly \emph{not} stable, and we flag it because the
instability is the finding: the Gaussian entanglement frontier holds only as
the memory approaches the information ceiling, with an error first order in
the fidelity gap \eqref{eq:conj1-limit}. Any statement of that result that
does not carry a fidelity qualifier is false in thin-memory regimes. The
episode also carries a methodological point that generalizes past this paper.
Two of the errors we found in our own first pass had the same shape: a
``regime'' sweep that varied a parameter which does not in fact move the
quantity being regimed, and a sign test gated by a rank rule rather than an
absolute threshold. Both made a result look stronger than it was, and neither
was visible from the output alone, but only from asking what the experiment
could have returned had the theory been false. A synthetic testbed makes
that question answerable, and Section~\ref{sec:what-can-fail} exists so that
a reader can check it row by row rather than taking our word for it. A
companion exercise replicating the identification tests on LLM agents with
genuine memory pipelines is the natural external-validity check; its design
is fixed by the interfaces used here.

All experiments, the environment, and the figure code are in the replication
package; every number in Table~\ref{tab:results} and every panel of
Figures~\ref{fig:regimes}--\ref{fig:pollution} is regenerated by three
scripts at a fixed seed, with all reported figures read from the committed
result files rather than transcribed. The full sweep runs in about five
minutes on one core. The calibration is disclosed rather than implicit: the
equilibrium and robustness experiments use a hand-selected operating point
($m=3$, $\sigma_{\nu0}=2.4$, $c_{S}=c_{P}=0.02$) chosen so that both efforts
are interior, since at unit costs every equilibrium effort corners at zero
and every equilibrium test is vacuous. One caveat for replicators: the
welfare optimum is an argmax over segments that are locally flat, so
platform-level differences of order $10^{-5}$ in $\bar Q$ can move summary
figures by one to three percentage points, so compare with a tolerance
rather than bitwise.

\section{Discussion and Limitations}\label{sec:limits}

The identification strategy that gives this paper its force is also the
source of most of its limits. By granting perfect provenance
(Assumption~\ref{as:provenance}) and complete information
(Assumption~\ref{as:completeinfo}) we bought the right to say that every
failure here survives the success of the infrastructure now being built. What
we gave up is any claim about the failures that arise \emph{because} that
infrastructure is imperfect: misattributed lineage, forged effort
attestations, buyers who cannot evaluate what they are buying. Those are real
and, on current evidence, quantitatively large. They are simply not what this
paper is about, and a reader should not treat our silence on them as a
judgment that they are second order. The same holds for the direction of our
conclusions: because the benchmark is generous to the market, our failure
results are, if anything, conservative, while our efficiency results are not.

The most consequential simplification is the scalar $\lambda$. Compressing
exposure into one number that enters utility linearly makes the rights
problem tractable and makes the ought--is comparison of
Proposition~\ref{prop:corners} reduce to two measurable quantities. It also
assumes what a great deal of the surrounding normative literature denies:
that the harm of being represented, deployed, and inferred from is the kind
of thing that trades against money at some rate. Dignitary harm, chilling
effects on what a person is willing to tell a system that remembers, and the
distinct injury of being \emph{accurately} known by a party one did not
choose are not obviously commensurable with a royalty. We have handled this
inside the model in the only way a transferable-utility benchmark permits:
inalienable uses enter as hard constraints rather than as priced quantities,
and Remark~\ref{rem:normative} flags the one place where pure efficiency
logic recommends tolerating subject distortion as a result to be scrutinized
rather than adopted. That is a containment strategy, not a solution. A model
in which some uses are lexically prior to compensation would change the
second-best bundle of Theorem~\ref{thm:secondbest}, and we do not know in
which direction.

Three structural restrictions bound the environment. First, there is one
platform: the portability results (Theorem~\ref{thm:rights}(ii) and
Corollary~\ref{cor:selfperp}) say that a monopolist will not supply $\phi$
and $\zeta$ voluntarily, which is a statement about a monopolist. Whether
competition between platforms substitutes for a portability mandate is
genuinely open, and it is the extension we would run first: the na\"ive
argument that competition forces portability runs into the observation that
low $\phi$ raises every incumbent's stake symmetrically, so the incentive to
withhold it is common rather than strategic. Second, the ``platform'' is a
single investor, whereas in practice the engineering margin is supplied by a
model provider, an application developer, and an infrastructure operator
whose contracts with each other are themselves incomplete;
Theorem~\ref{thm:rights} is a two-investor result and the budget-balance
obstruction only sharpens with more parties, but the \emph{allocation}
question does not have an obvious extension. Third, memory has one subject:
Definition~\ref{def:artifact}'s person-linkage condition is stated for a
single individual, and Section~\ref{sec:model} excludes multi-subject
artifacts by construction. This is the exclusion we are least comfortable
with, because most real interaction history describes more than one person:
a conversation, a household, a clinical encounter. Theorem~\ref{thm:formation}
supplies the machinery for the externality (relational leakage onto
non-listing peers is precisely the third wedge) but not the rights
allocation, and no single subject's consent can authorize an artifact whose
subject is plural.

Two modelling choices are assumptions we would like to see tested rather than
limitations we can bound. We assume exposure rises with fidelity: an accurate
memory is a more dangerous one. Remark~\ref{rem:misrep} shows that the
opposite case, harm from \emph{mis}representation, adds a term that
pushes the subject's refresh incentive up rather than down, which would
weaken the depletion mechanism without reversing it, but the two harms
plainly coexist and their relative magnitude is an empirical matter that our
sign restriction settles by assumption. The second-best bundle of
Section~\ref{sec:optimal} is also computed under commitment
(Assumption~\ref{as:commitment}). Without it the coverage floor is
renegotiation-proof only in the region where Theorem~\ref{thm:secondbest}(ii)
makes it self-enforcing, which is exactly the region where the subject is a
productive refresher; the commitment assumption is therefore doing real work
outside that region, and the delegation results of
Section~\ref{sec:delegation} are a partial rather than a complete substitute
for it.

Finally, three propositions remain conjectural in different senses. The
Gaussian entanglement frontier is now understood, not assumed: the testbed
established that $\kappa^{\star}=\chi$ is an information-saturation limit
with error first order in the fidelity gap \eqref{eq:conj1-limit}. What
remains open is whether that rate survives outside the Gaussian-linear
environment, where ``one direction of the personal block'' has no clean
meaning; we would expect the qualitative conclusion (targeted garbling
dominates identity replacement, increasingly so in the dimension of identity)
to be robust, and the constant not to be. Sophistication selection, the
conjecture that the subjects most able to evaluate a rights bundle are also
those with the most valuable memory, would if true make
Proposition~\ref{prop:selection}'s composition effect worse than we report
and is straightforwardly testable on any live marketplace. And the no-fusion
rule of Section~\ref{sec:nofusion} establishes that provenance, performance,
and permission must not be collapsed into one score; it does not establish
that three credentials suffice. Lemma~\ref{lem:no-score} rules out one
dimension, not $k$, and we do not have a lower bound on credential
dimensionality. Section~\ref{sec:threats} states the separate question of
what the synthetic testbed can and cannot license; the LLM-agent replication
described there is the natural next external-validity check.

It is worth naming what would overturn the paper rather than merely qualify
it. The entire ought--is apparatus rests on the subject's refresh margin
being live: on $k_{S}$ being non-trivial in real systems. If the empirical
answer is that engineering effort dominates and users contribute almost
nothing to the continued validity of their own memory, then
Theorem~\ref{thm:headline}'s depletion channel is vacuous, the scope haircut
vanishes with the margin it protects (which is exactly what our own
dead-margin cells show), and platform ownership becomes not merely the de
facto arrangement but the efficient one. We do not believe that is the world
we are in, and Section~\ref{sec:prereg}'s first design is built to find out.
But it is the finding that would falsify us, and it is measurable now.

\section{Conclusion}\label{sec:conclusion}

Provenance is necessary infrastructure, but it is not a complete market
institution. Under \emph{perfect} provenance and complete information (the
world the builders of memory marketplaces are trying to reach, granted here
in full), the market for personalized AI memory still fails, and it fails in
structured, identifiable, correctable ways. Verified production effort does
not imply decision value. No budget-balanced assignment of cash-flow rights
makes both co-producers residual claimants. Widening reach can raise access
today and shrink the stock tomorrow. Audits that replace a person's identity
overstate how entangled the artifact really is, by more the richer that
identity becomes. None of these is a verification failure, and none is
repaired by better attestation.

Behind all of them is one mechanism: use can be non-rival while refresh is
relationally supplied. Copying a memory costs nothing; keeping it true
requires the continued participation of the person it describes, who bears
exposure from every deployment. An asset with a nonrival use and a rival
maintenance technology does not behave like data, and the institutions that
work for data do not transfer to it.

That mechanism also resolves the public argument this paper took as its
organizing question. Memory \emph{ought} to belong to its subject; memory
\emph{is} held by platforms. Both positions are right about what they see and
both are answers to the wrong question, because they are twin corners of one
rights space and each extinguishes a different side of the asset's
production: whoever does not own supplies no refresh. The efficient
arrangement is interior, the de facto corner is self-perpetuating through
precisely the parameters (portability and salience) that determine whether
the normative corner is even reachable, and reform that renames the owner
while leaving those parameters in place changes neither scope nor refresh.
The regulatory question is not who owns the memory but what claim structure
the subject holds, and whether the conditions for exercising it exist.

What follows from that is buildable now: credentials that keep provenance,
performance, and permission separate rather than fusing them into a single
trust score; compensation indexed to measured fidelity rather than flat,
which is what shuts down adverse selection into authorization; licenses
scoped by deployment class rather than sold outright; renewable rather than
perpetual claims, because an asset that must be maintained cannot be
alienated once; staleness reported per use class, since the same artifact is
fresh for one buyer and stale for another; entanglement reported against a
named adversary rather than as an unqualified number; and portability and
salience treated as mandates rather than as features an incumbent might
choose to supply. Each of these carries a theorem pointer in
Section~\ref{sec:governance}, and each is cheaper to require before the asset
class matures than after.

That timing is the practical point: in the 23andMe bankruptcy, deletion
rights over genetic data derived from millions of people were negotiated by
state attorneys general only after the asset had entered the estate.
Personalized memory is on the same trajectory and is harder in two ways: it
is continuously re-created rather than collected once, and the person it
describes keeps producing it. The institutions are cheap to specify today and
expensive to retrofit; the window is open now.

\appendices

\section{Proof of Theorem~\ref{thm:rights}}\label{app:rights}

\emph{Part (i).} Work in the additive benchmark of
Lemma~\ref{lem:refresh}; the argument for $\gamma\neq0$ is identical because
the planner's problem replaces both stakes by $\Omega$ in the same system.
Given stakes $(w_{S},w_{P})$, the effort subgame has first-order conditions
$c_{S}e=w_{S}\eta_{S}$ and $c_{P}x=w_{P}\eta_{P}$, so
$e^{E}=\eta_{S}w_{S}/c_{S}$ and $x^{E}=\eta_{P}w_{P}/c_{P}$. The planner
maximizes $\Omega Y-c_{S}e^{2}/2-c_{P}x^{2}/2$, whose first-order conditions
are $c_{S}e=\Omega\eta_{S}$ and $c_{P}x=\Omega\eta_{P}$. Since
$\eta_{S},\eta_{P}>0$ and both problems are strictly concave in own effort,
$e^{E}=e^{FB}$ if and only if $w_{S}=\Omega$, and $x^{E}=x^{FB}$ if and only
if $w_{P}=\Omega$. First-best refresh on both margins therefore requires
$w_{S}=w_{P}=\Omega$, hence $w_{S}+w_{P}=2\Omega$.

Ex post budget balance with no external subsidy restricts the bundle to
allocate the realized private cash flow and the monetizable exposure, so by
\eqref{eq:wS}--\eqref{eq:Omega} the stakes satisfy
$w_{S}+w_{P}=T=\Omega-C$. Because buyer surplus is nonnegative, $C\ge0$ and
therefore $T\le\Omega<2\Omega$ whenever $\Omega>0$, which the assumption of a
strictly positive first best guarantees. The two requirements are
incompatible, and the conclusion does not depend on how the transfer $F$ is
set, since $F$ moves rents and not marginal claims. Perfect price
discrimination is the tightest case: it sets $C=0$ and still leaves
$w_{S}+w_{P}=\Omega<2\Omega$. Any scalar owner is the special case
$w_{S}\in\{0,T\}$, so ownership assignment cannot restore first best either.
The result is Holmstr\"om-type team moral hazard with one memory-specific
amplification: the planner's weight $\Omega$ contains buyer surplus that no
budget-balanced private bundle appropriates, so the wedge is at least $C$
even before any incentive conflict between the parties.

\emph{Part (ii).} Fix $(\alpha,\tau)$ and write
$\widehat Q(n)=(1-\delta)Q_{t}+Y^{E}(n)$ with
$Y^{E}=k_{S}w_{S}(n)+k_{P}w_{P}(n)$. Suppose the platform names scope at
stage 2, anticipating stage-3 play. Its objective is
$\pi_{P}(n)=w_{P}(n)\widehat Q(n)-c_{P}(x^{E})^{2}/2$ with
$x^{E}=\eta_{P}w_{P}/c_{P}$. Differentiating,
\begin{equation*}
\begin{split}
\pi_{P}'&=w_{P}'\widehat Q
+w_{P}\bigl(k_{S}w_{S}'+k_{P}w_{P}'\bigr)
-c_{P}x^{E}\frac{dx^{E}}{dn}\\
&=w_{P}'\widehat Q+k_{S}w_{P}w_{S}',
\end{split}
\end{equation*}
because $c_{P}x^{E}(dx^{E}/dn)=k_{P}w_{P}w_{P}'$ cancels the own-effort term
in $\widehat Q'$. This is an envelope argument: the platform's own effort is
already optimal, so only the \emph{cross} effect of scope on the subject's
effort survives. The symmetric computation for a subject-controlled scope
gives $\pi_{S}'=w_{S}'\widehat Q+k_{P}w_{S}w_{P}'$. An interior scope
therefore solves
\begin{equation}\label{eq:controlfocs}
w_{P}'\widehat Q+k_{S}w_{P}w_{S}'=0
\quad\text{or}\quad
w_{S}'\widehat Q+k_{P}w_{S}w_{P}'=0
\end{equation}
according to who controls it. The two conditions coincide at a common $n$
only if
$\widehat Q\,(w_{P}'-w_{S}')=k_{P}w_{S}w_{P}'-k_{S}w_{P}w_{S}'$, a single
equation in the primitives, so they differ on a set of full measure in
$(\alpha,\tau,\lambda,k_{S},k_{P})$. The corner cases $\alpha\in\{0,1\}$ are
handled by evaluating \eqref{eq:controlfocs} with the corresponding stake set
to its bound, which removes one term without equating the two conditions.

For the sharper statement, specialize to the fixed-stock benchmark of
Proposition~\ref{prop:scope}, where $\widehat Q$ is a constant and the two
conditions reduce to $w_{P}'=0$ and $w_{S}'=0$. With $\tau=0$ and
$0<\alpha<1$ these give $n_{P}=1/2$ and
$n_{S}=[1-\lambda/\alpha]_{[0,1]}/2$. Varying $\alpha$ moves $n_{S}$ but
leaves $n_{P}$ at $1/2$, while changing the identity of the controller moves
scope at fixed $\alpha$. Cash-flow rights and scope control therefore act on
different margins and neither substitutes for the other.

Two scope conditions of the theorem's proof are worth recording as limits of
its reach. It excludes only budget-balanced implementations that operate
through marginal claims on realized refresh. Verifiable effort, an external
subsidy financed outside the pair, or a performance contract written on a
signal other than $Y$ all enlarge the domain and can restore first best;
none of them is available under Assumption~\ref{as:provenance} alone, which
verifies lineage rather than effort.

\section{Proof of Theorem~\ref{thm:formation}(ii)--(iv)}\label{app:formation}

Throughout, $G$ denotes the participation mass at the prevailing threshold,
$g$ its density, and $P(\hat\lambda)$ the marginal relational harm defined in
the statement, normalized so that
$\tfrac{d}{d\hat\lambda}\bigl[\Lambda(\varrho)G(\hat\lambda)
\bar\lambda^{\mathrm{nl}}(\hat\lambda)\tilde qD\bigr]
=g(\hat\lambda)\,\Lambda(\varrho)P(\hat\lambda)$.

\emph{Part (ii).} Differentiating the welfare functional in part (iv) of the
statement,
\begin{equation}\label{eq:Wprime}
W'(\hat\lambda)=g(\hat\lambda)\Bigl[v\Theta(G)-\hat\lambda
+v\Theta'(G)G-\Lambda(\varrho)P\Bigr],
\end{equation}
where the third term collects the effect of the marginal entrant on every
inframarginal participant's thickness and the fourth collects the relational
harm she imposes on non-listing peers. Setting the bracket to zero gives
\eqref{eq:planner-threshold}. The private cutoff instead equates the
lister's own compensated benefit to her own exposure,
$\hat\lambda=sv\Theta(G)$, which is \eqref{eq:market-threshold}.

To compare the two fixed points rather than the two conditions, define the
excess functions
\begin{equation*}
\Phi^{M}(\hat\lambda):=sv\Theta(G(\hat\lambda))-\hat\lambda,\qquad
\Phi^{W}(\hat\lambda):=W'(\hat\lambda)/g(\hat\lambda),
\end{equation*}
whose zeros are the market and planner thresholds. Subtracting,
\begin{equation}\label{eq:excessgap}
\Phi^{M}(\hat\lambda)-\Phi^{W}(\hat\lambda)
=\Lambda(\varrho)P-(1-s)v\Theta-v\Theta'(G)G
\end{equation}
at every $\hat\lambda$, so the right side is positive exactly when
\eqref{eq:wedges} holds. Assume the regularity condition
\begin{equation}\label{eq:singlecross}
sv\,\Theta'(G(\hat\lambda))\,g(\hat\lambda)<1
\quad\text{for all }\hat\lambda\text{ in the support of }G,
\end{equation}
which is the negation of the multiplicity condition in part (i) and makes
$\Phi^{M}$ strictly decreasing, and assume the analogous bound for
$\Phi^{W}$. Each excess function then crosses zero at most once and from
above, so the pointwise ordering in \eqref{eq:excessgap} transfers to the
fixed points: $\hat\lambda^{*}>\hat\lambda^{W}$ if and only if
\eqref{eq:wedges} holds. Without \eqref{eq:singlecross} the market has
multiple thresholds, and the comparison should be read equilibrium by
equilibrium or after the global-games selection of part (i).

\emph{Part (iii).} By Definition~\ref{def:stock}, thickness enters through
$\Theta=\Psi(M)$ with
$\Psi'(M)=\bar\psi/(M\bar\psi\sigma_{0}^{2}+1)^{2}$, so $\Psi'(M)\to0$ as
$M\to\infty$ and $\Theta$ converges to a finite ceiling
$\Theta_{\infty}$. Hence $v\Theta'(G)G\to0$ along any sequence of thickening
markets, while $(1-s)v\Theta\to(1-s)v\Theta_{\infty}$, which vanishes as
coverage approaches full appropriation. The left side of \eqref{eq:wedges}
does not attenuate: for $\varrho>0$ and non-degenerate
$\bar\lambda^{\mathrm{nl}}$ there is $\underline{P}>0$ with
$\Lambda(\varrho)P\ge\Lambda(\varrho)\underline{P}>0$ uniformly, because the
correlation that generates leakage is a property of the class rather than of
its size. Therefore for every $\varrho>0$ and every $s$ sufficiently close to
one there is a finite thickness beyond which \eqref{eq:wedges} holds, and
every further lister is socially excessive. The naive-participation statement
follows by noting that a subject with salience $\zeta<1$ applies the cutoff
$\zeta\lambda_{i}<\lambda_{i}$, which shifts $\hat\lambda^{*}$ to the right,
while $W$ continues to be evaluated at the true $\lambda_{i}$, so the wedge
in \eqref{eq:excessgap} widens at unchanged welfare weights.

\emph{Part (iv).} Non-monotonicity is immediate from \eqref{eq:Wprime}: near
$\hat\lambda=0$ the bracket is $v\Theta(0)-\Lambda(\varrho)P(0)$, which is
positive whenever thickness value dominates leakage at the bottom of the
exposure distribution, and the bracket is strictly decreasing under the
regularity condition, so $W$ rises and then falls once \eqref{eq:wedges}
binds. For the formation statement, integrate \eqref{eq:Wprime} from the
no-market point, at which $W(0)=0$:
\begin{equation*}
W(\hat\lambda^{*})=\int_{0}^{\hat\lambda^{*}}
\Bigl[v\Theta(G)-\ell+v\Theta'(G)G-\Lambda(\varrho)P\Bigr]g(\ell)\,d\ell .
\end{equation*}
Voluntary formation is welfare-negative precisely when this integral is
negative, which occurs when the bracket is negative over a $g$-weighted
majority of the participating range. Three sufficient configurations follow
directly. First, $\Lambda(\varrho)$ large enough that
$\Lambda(\varrho)P(0)>v\Theta(0)$ makes the bracket negative from the start.
Second, mass of $G$ concentrated at high $\ell$ makes $-\ell$ dominate over
most of the range. Third, a large naive share raises $\hat\lambda^{*}$ into
the region where the bracket has already turned negative. In each case
marginal growth subsidies reduce welfare. The reason the conclusion can
contradict participants' revealed preference is structural rather than
behavioral: relational externalities fall on parties who never chose, and
naive participation is not informed revelation.

\section{Proof of Theorem~\ref{thm:secondbest}(iii)--(iv)}\label{app:secondbest}

Work in the uniform benchmark $R(n)=n(1-n)$, $C(n)=n^{2}/2$,
$B(n)=n-n^{2}/2$ with $u_{S}=u_{P}=0$, and write $m:=1-\lambda$ for the
static welfare scope, so that
\begin{equation}\label{eq:TOm}
T(n)=n(m-n),\qquad \Omega(n)=\tfrac{n(2m-n)}{2},\qquad C=\Omega-T .
\end{equation}
Both are positive on $n\in(0,m)$, which is the active region. By
Theorem~\ref{thm:secondbest}(i) optimal coverage implements
Corollary~\ref{cor:stake}'s profile, so with $k_{S}=a$, $k_{P}=b$ the
concentrated objective is $W(n)=\max_{w\in[0,T]}V(n,w)$ with
\begin{equation*}
V(n,w)=\Omega\bigl[aw+b(T-w)\bigr]-\tfrac{a w^{2}}{2}
-\tfrac{b(T-w)^{2}}{2},
\end{equation*}
which is strictly concave in $w$ with unconstrained maximizer
$w^{\mathrm{unc}}=[(a-b)\Omega+bT]/(a+b)$, so
$w_{S}^{\star}=\operatorname{Proj}_{[0,T]}w^{\mathrm{unc}}$ as claimed. Write
$\kappa:=a/b$. Since $w^{\mathrm{unc}}\ge0$ is equivalent to
$\kappa\ge C/\Omega$ and $w^{\mathrm{unc}}\le T$ to $\kappa\le\Omega/C$, the
stake is interior exactly on
\begin{equation}\label{eq:interiorband}
\frac{n}{2m-n}\;\le\;\kappa\;\le\;\frac{2m-n}{n},
\end{equation}
using $C/\Omega=n/(2m-n)$ from \eqref{eq:TOm}. This is the same inequality as
the self-enforcing floor \eqref{eq:floor}, which is why parts (ii) and (iv)
are two readings of one condition.

\emph{Interior branch.} Substituting $w^{\mathrm{unc}}$ into $V$ and
simplifying,
\begin{equation}\label{eq:Wint}
\begin{split}
W_{\mathrm{int}}(n)=\frac{n^{2}}{8(a+b)}\Bigl[&(a^{2}+b^{2})(2m-n)^{2}\\
&+2ab\bigl(2m^{2}-4mn+n^{2}\bigr)\Bigr],
\end{split}
\end{equation}
a quartic in $n$ vanishing at $n=0$. Differentiating and dividing by the
positive factor $n(a+b)/2$,
\begin{equation}\label{eq:foc-int}
\begin{split}
&W_{\mathrm{int}}'(n)\;\propto\;n^{2}-3mn+2m^{2}\theta(\kappa),\\
&\theta(\kappa):=\frac{\kappa^{2}+\kappa+1}{(\kappa+1)^{2}} .
\end{split}
\end{equation}
The quadratic in \eqref{eq:foc-int} opens upward with roots
$n^{\pm}=\tfrac{m}{2}\bigl(3\pm\sqrt{9-8\theta}\bigr)$. Because
$\theta(\kappa)\in[3/4,1)$ for all $\kappa>0$, with the minimum at
$\kappa=1$, we have $\sqrt{9-8\theta}\in(1,\sqrt3\,]$ and therefore
$n^{+}\ge2m>m$, which is infeasible, while $n^{-}\in[\,(3-\sqrt3)m/2,\,m)$.
Hence $W_{\mathrm{int}}$ is strictly increasing on $(0,n^{-})$ and strictly
decreasing on $(n^{-},m)$, so it is unimodal with the unique interior
maximizer $n^{-}$. Setting $\kappa=1$ gives $\theta=3/4$,
\begin{equation*}
W_{\mathrm{int}}(n)=\tfrac{k}{4}\,n^{2}(m-n)(3m-n),\qquad
n^{-}=\tfrac{3-\sqrt3}{2}\,m,
\end{equation*}
which is part (iii). Optimal coverage follows from \eqref{eq:sstar}: at
$\kappa=1$ the stake is $w_{S}^{\star}=T/2$, so
$s^{*}=\lambda+T/(2n)=\lambda+(m-n)/2$, and evaluating at $n^{-}$ gives
$s^{*}=\lambda+\tfrac{\sqrt3-1}{4}m$. Finally
$m/2<(3-\sqrt3)m/2<m$ places $n^{*}$ strictly between the subject-corner
scope and the static welfare scope.

\emph{Cornered branch.} If the projection binds, one stake is zero and the
other is $T$; write $k_{j}$ for the productive party's coefficient. Then
$Y^{E}=k_{j}T$ and the effort cost is $k_{j}T^{2}/2$, so
\begin{equation*}
W_{\mathrm{cor}}(n)=k_{j}\Bigl[\Omega T-\tfrac{T^{2}}{2}\Bigr]
=\tfrac{k_{j}m}{2}\,n^{2}(m-n),
\end{equation*}
using \eqref{eq:TOm}. Its derivative is proportional to $n(2m-3n)$, so
$W_{\mathrm{cor}}$ is unimodal on $(0,m)$ with maximizer $n=\tfrac{2}{3}m$,
which is part (iv)'s polar case.

\emph{The band is exact and attained.} The two branches are not alternatives
to be compared globally, since \eqref{eq:interiorband} determines which one
is operative at each $n$, and the switch points depend on $\kappa$. Evaluate
\eqref{eq:interiorband} at $n=\tfrac{2}{3}m$: the bounds are $1/2$ and $2$.
Hence for $\kappa\in(1/2,2)$ the stake is interior at the cornered branch's
own maximizer, so the operative solution is the interior one; and for
$\kappa\le1/2$ or $\kappa\ge2$ the projection binds there and the operative
solution is the cornered one. The two descriptions agree at the boundary: at
$\kappa=1/2$ (and by symmetry at $\kappa=2$) we get $\theta=7/9$,
$\sqrt{9-8\theta}=5/3$ and $n^{-}=\tfrac{2}{3}m$, so the concentrated program
and its maximizer are continuous in $\kappa$. Collecting,
\begin{equation}\label{eq:nstar-general}
\frac{n^{*}}{1-\lambda}=
\begin{cases}
\tfrac{1}{2}\bigl(3-\sqrt{9-8\theta(\kappa)}\bigr), & \tfrac12\le\kappa\le2,\\[4pt]
\tfrac{2}{3}, & \text{otherwise},
\end{cases}
\end{equation}
and since $\theta$ ranges over $[3/4,7/9]$ on $[1/2,2]$, the ratio ranges
over $\bigl[\tfrac{3-\sqrt3}{2},\tfrac{2}{3}\bigr]$, attaining its minimum at
$\kappa=1$ and its maximum at every $\kappa\notin(1/2,2)$. This proves
\eqref{eq:band} as an exact and attained band rather than a numerical
finding, and it supersedes the regularity condition under which the theorem
was originally stated: the concentrated program is unimodal on $(0,1-\lambda)$
for every productivity ratio, so no quasiconcavity assumption is needed. The
monotonicity in $\lambda$ is immediate, because $\lambda$ enters
\eqref{eq:nstar-general} only through the factor $m=1-\lambda$.

The economic content of \eqref{eq:nstar-general} is that the haircut is
almost insensitive to \emph{who} refreshes and sensitive only to \emph{that}
someone does. The ratio moves by less than four percentage points across the
entire range of productivity ratios, which is why
Remark~\ref{rem:two-inputs}'s claim about estimation requirements holds: an
analyst who knows only $\lambda$ can set scope nearly optimally, whereas
coverage requires the split.

\section{Proof of Proposition~\ref{prop:regimes}}\label{app:regimes}

On the polluted region $s<\lambda$ the aware subject distorts by
$d^{*}=(\lambda-s)nQ/c_{d}$ from Theorem~\ref{thm:pollution}(i), and receipts
are scaled by effective quality $\tilde q(s)=1-d^{*}(s)$, so the objective
\eqref{eq:program} becomes
\begin{equation*}
W^{-}(n,s)=\tilde q(s)\,\Omega(n)\,Y^{E}(s)
-\tfrac{k_{S}w_{S}^{2}}{2}-\tfrac{k_{P}w_{P}^{2}}{2},
\end{equation*}
with $\tilde q'(s)=nQ/c_{d}>0$, while for $s\ge\lambda$ we have
$\tilde q\equiv1$ and $W^{+}=W^{E}$. The objective is continuous at
$s=\lambda$ and kinked there. Using
$\partial w_{S}/\partial s=n$, $\partial w_{P}/\partial s=-n$ and
$\partial Y^{E}/\partial s=n(k_{S}-k_{P})$, and evaluating at $s=\lambda$
where $w_{S}=0$, $w_{P}=T$ and $Y^{E}=k_{P}T$,
\begin{align}
\left.\frac{\partial W}{\partial s}\right|_{s=\lambda^{+}}
&=n\bigl[(k_{S}-k_{P})\Omega+k_{P}T\bigr]
=n(k_{S}+k_{P})\,w^{\mathrm{unc}},\label{eq:above}\\
\left.\frac{\partial W}{\partial s}\right|_{s=\lambda^{-}}
&=n\Bigl[(k_{S}-k_{P})\Omega+k_{P}T
+\tfrac{Q}{c_{d}}\,\Omega\,k_{P}T\Bigr].\label{eq:below}
\end{align}
The extra term in \eqref{eq:below} is the pollution-relief value of raising
coverage: a marginal increase in $s$ lowers the distortion and so raises
$\tilde q$, which is worth $\tilde q'\Omega Y^{E}=(nQ/c_{d})\Omega k_{P}T$.

Regime (a) follows from \eqref{eq:above}: the derivative from above is
positive, so the optimum lies strictly inside the no-pollution region,
exactly when $w^{\mathrm{unc}}>0$, which by \eqref{eq:interiorband} is
$k_{S}/k_{P}>n/a_{2}$ with $a_{2}=2-2\lambda-n=2\Omega/n$. Regime (b) is the
boundary case $w^{\mathrm{unc}}=0$, at which the subject is exactly covered
and holds no net claim.

Regime (c) requires the derivative from below to be negative at the kink, so
that welfare rises as coverage falls below the floor. By \eqref{eq:below}
this is
\begin{equation*}
(k_{P}-k_{S})\Omega-k_{P}T>\tfrac{Q}{c_{d}}\,\Omega\,k_{P}T .
\end{equation*}
Substituting $\Omega=na_{2}/2$ and $T=na_{1}$ with $a_{1}=1-\lambda-n$ and
multiplying by $2c_{d}/n$ gives exactly \eqref{eq:tolerate}. Equivalently,
\begin{equation*}
(k_{S}+k_{P})\,w^{\mathrm{unc}}\;<\;-\tfrac{Q}{c_{d}}\,\Omega\,k_{P}T,
\end{equation*}
so tolerated pollution requires the unconstrained stake to be not merely
negative but negative by more than the pollution-relief value. This is the
formal content of the reading given in the statement: the platform must be
the strongly dominant refresher, so that $w^{\mathrm{unc}}$ is far below
zero, \emph{and} the distortion technology must be stiff, so that $c_{d}$ is
large and the relief term is small. Either condition alone is insufficient,
which is why regime (c) is a corner of the parameter space rather than the
generic case.

\section{General private flows}\label{app:flows}

Section~\ref{sec:optimal} sets $u_{S}=u_{P}=0$. Restoring them changes the
accounting but not the structure. From \eqref{eq:wS}--\eqref{eq:Omega},
\begin{equation*}
T(n)=u_{S}+u_{P}+R(n)-\lambda n,\qquad \Omega(n)=T(n)+C(n),
\end{equation*}
and the instrument collapse of Proposition~\ref{prop:collapse}(i) becomes
$w_{S}=u_{S}+(s-\lambda)n$ and $w_{P}=u_{P}+R(n)-sn$, so coverage remains the
single static instrument. Corollary~\ref{cor:stake}'s stake profile is
unchanged as a statement about $w_{S}^{\star}$, because it is derived at
fixed $T$ and does not reference how $T$ is composed. The only substantive
change is the map from the optimal stake to the optimal coverage: instead of
\eqref{eq:sstar} we obtain
\begin{equation}\label{eq:sstar-general}
s^{*}(n)=\lambda+\frac{w_{S}^{\star}(n)-u_{S}}{n}.
\end{equation}
A subject who already receives a private flow $u_{S}$ from the relationship
needs proportionally less explicit coverage to reach the same stake, and the
authenticity floor $s^{*}\ge\lambda$ now holds if and only if
$w_{S}^{\star}\ge u_{S}$ rather than $w_{S}^{\star}\ge0$. The floor is
therefore harder to clear when the subject's outside private flow is large,
which is the formal version of a familiar objection: a platform can point to
the free service a user already enjoys as a reason to pay less, and the model
agrees arithmetically while noting that the substitution leaves the stake,
and hence refresh, exactly where it was. The comparative statics of
Theorem~\ref{thm:secondbest} carry over with $T$ and $\Omega$ redefined; the
closed forms of parts (iii) and (iv) are specific to $u_{S}=u_{P}=0$ and to
uniform demand, and \eqref{eq:nstar-general} should be recomputed for other
primitives.

\section{Information extensions}\label{app:info}

Assumption~\ref{as:completeinfo} is maintained throughout the body. This
appendix records what relaxing it adds, and why we regard the additions as
the least memory-specific part of the problem.

\emph{Hidden fidelity.} Suppose the seller observes $\tilde q$ and the buyer
observes only its class distribution. Since willingness to pay is
$A_{b}q_{u}$ and use-indexed fidelity is increasing in $\tilde q$, a uniform
price prices the conditional mean, high-fidelity artifacts withdraw, and the
pool unravels in the manner of \cite{akerlof1970}. Two features of our object
change the repair rather than the diagnosis. First, the certifiable object is
not authenticity but \emph{performance on a named use class}, because
Lemma~\ref{lem:no-score} rules out a single scalar quality index; the
certificate must therefore be indexed by buyer, task, and host model, which
is a richer and more expensive credential than a provenance attestation.
Second, a monopoly certifier has the Lizzeri incentive to coarsen what it
reveals \cite{lizzeri1999}, and coarsening is especially damaging here
because it collapses precisely the use-indexing that makes the certificate
informative. This is a second and independent argument for the no-fusion rule
of Section~\ref{sec:nofusion}.

\emph{Evaluation leaks the asset.} A buyer who must assess fit before paying
learns some of what he would be paying for, which is the disclosure problem
of \cite{antonyao1994,antonyao2002}. The severity is indexed by the
entanglement coefficient of Definition~\ref{def:kappa}: for low
$\kappa^{\star}$ classes a garbled sample is nearly as informative as the
artifact and leaks little, whereas for high $\kappa^{\star}$ classes any
sample informative about fit is also informative about the person. The
implied market design is therefore class-dependent, moving from open samples
at low $\kappa^{\star}$ to query sandboxes and audited replay at high
$\kappa^{\star}$, where the buyer evaluates through an interface that returns
task performance without releasing the state.

\emph{What does not change.} The results of
Sections~\ref{sec:results}--\ref{sec:optimal} are statements about
distortions that survive symmetric information, so none of them is repaired
by these additions and none is driven by them. The certifier's economically
correct role was already pinned down under complete information in
Remark~\ref{rem:why-tau}: make $\tilde q$ contractible so that coverage can
be quality-indexed, which is what shuts down the selection effect of
Proposition~\ref{prop:selection}. Certification serves authorization and does
not substitute for it, and that ordering is unaffected by who knows what.

\bibliographystyle{IEEEtran}
\bibliography{references}

\end{document}